\documentclass[12pt,letterpaper]{article}
\usepackage[letterpaper]{geometry}
\usepackage{graphicx}
\usepackage{natbib}
\usepackage{tikz}
\usepackage{multirow}
\usepackage{dcolumn,booktabs}
\usepackage{algorithm,algorithmicx,
  algpseudocode,algcompatible}
\usepackage{url}
\usepackage{enumitem}
\usepackage[font=footnotesize]{caption} 
\usepackage{amsmath,amssymb,bm,amsthm}
\newtheorem{theorem}{Proposition}
\newtheorem{lemma}{Lemma}
\algrenewcommand{\algorithmiccomment}[1]{
  \hfill$\blacktriangleright$ #1}
\newcommand{\bigcell}[2]{\begin{
  tabular}{@{}#1@{}}#2\end{tabular}}

\DeclareMathOperator{\rank}{rank}
\DeclareMathOperator{\vecmat}{vec}

\newcommand{\q}{q}

\newcommand{\hsigma}{\hat{\sigma}}
\newcommand{\ha}{\hat{a}}

\newcommand{\hx}{\hat{x}}
\newcommand{\hX}{\widehat{X}}

\newcommand{\bzero}{\boldsymbol 0}
\newcommand{\bone}{\boldsymbol 1}

\newcommand{\bha}{\hat{\boldsymbol{a}}}

\newcommand{\bof}{\boldsymbol f}
\newcommand{\bol}{\boldsymbol g}

\newcommand{\bu}{\boldsymbol u}
\newcommand{\btu}{\boldsymbol{\tilde{u}}}
\newcommand{\bv}{\boldsymbol v}
\newcommand{\btv}{\boldsymbol{\tilde{v}}}
\newcommand{\bx}{\boldsymbol x}

\newcommand{\bA}{\boldsymbol A}

\newcommand{\bD}{\boldsymbol D}

\newcommand{\bR}{\boldsymbol R}
\newcommand{\bS}{\boldsymbol S}

\newcommand{\bU}{\boldsymbol U}
\newcommand{\btU}{\boldsymbol{\widetilde{U}}}
\newcommand{\bhU}{\boldsymbol{\widehat{U}}}
\newcommand{\bV}{\boldsymbol V}
\newcommand{\btV}{\boldsymbol{\widetilde{V}}}
\newcommand{\bhV}{\boldsymbol{\widehat{V}}}
\newcommand{\bW}{\boldsymbol W}
\newcommand{\btW}{\boldsymbol{\widetilde{W}}}
\newcommand{\bX}{\boldsymbol X}
\newcommand{\bhX}{\boldsymbol{\widehat{X}}}

\newcommand{\btheta}{\bm \theta}

\newcommand{\bhmu}{\boldsymbol{\hat{\mu}}}

\newcolumntype{M}[1]{>{\centering\arraybackslash}m{#1}}

\definecolor{orange1}{RGB}{255,128,0}
\definecolor{purple2}{RGB}{102,0,204}
\definecolor{blue}{RGB}{0,0,255}
\definecolor{red}{RGB}{255,0,0}

\begin{document}

\def\spacingset#1{\renewcommand{\baselinestretch}
{#1}\small\normalsize} \spacingset{1}


\phantom{abc}

\vspace{5mm}

\begin{center}
\LARGE{\bf Multivariate Singular Spectrum Analysis\\ by  
  Robust Diagonalwise Low-Rank Approximation}\\
\vspace{10mm}
\large{Fabio Centofanti\footnote[1]{
        Department of Industrial Engineering, University
				of Naples Federico II, Naples, Italy,
				\texttt{fabio.centofanti@unina.it}}, 
        Mia Hubert\footnote[2]{
				Section of Statistics and Data Science, Department 
        of Mathematics, KU Leuven, Belgium, \texttt{
				$\{$mia.hubert,peter.rousseeuw$\}$@kuleuven.be}}, 
        Biagio Palumbo$^1$, and Peter J. Rousseeuw$^2$}\\
\vspace{8mm}
September 30, 2023\\
\end{center}
\vspace{3mm}

\bigskip
\begin{center}
{\bf Abstract}
\end{center}
Multivariate Singular Spectrum Analysis (MSSA) is a powerful 
and widely used nonparametric method for multivariate time 
series, which allows the analysis of complex temporal data 
from diverse fields such as finance, healthcare, ecology, 
and engineering.
However, MSSA lacks robustness against outliers because it 
relies on the singular value decomposition, which is very 
sensitive to the presence of anomalous values. MSSA can 
then give biased results and lead to erroneous conclusions.
In this paper a new MSSA method is proposed, named 
\textit{RObust Diagonalwise Estimation of SSA} (RODESSA), 
which is robust against the presence of cellwise and 
casewise outliers. In particular, the decomposition step 
of MSSA is replaced by a new robust low-rank approximation 
of the trajectory matrix that takes its special structure 
into account. 
A fast algorithm is constructed, and it is proved that 
each iteration step decreases the objective function. 
In order to visualize different types of outliers, a new 
graphical display is introduced, called an enhanced time 
series plot.
An extensive Monte Carlo simulation study is performed to 
compare RODESSA with competing approaches in the literature. 
A real data example about temperature analysis in passenger 
railway vehicles demonstrates the practical utility of the 
proposed approach.\\

\vspace{10mm}


\spacingset{1.5}

\newpage
\section{Introduction} \label{sec:intro}

Time series analysis plays a crucial role in 
understanding and predicting the behavior of 
sequential data across a wide range of disciplines. 
It has proven to be an invaluable tool in fields such 
as finance, healthcare, ecology, engineering, and more. 
Singular spectrum analysis (SSA) has emerged as a 
powerful nonparametric tool for extracting valuable 
insights from time-dependent data. A comprehensive 
overview of SSA can be found in the books 
\cite{golyandina2001analysis}, 
\cite{golyandina2013singular}, 
and \cite{golyandina2018singular}.
Numerous examples showcasing the success of SSA can 
be found in the literature. Multivariate singular 
spectrum analysis, referred to as MSSA 
\citep{broomhead1986qualitative}, enables the 
simultaneous analysis and interpretation of multiple 
time series, by exploiting the dependencies between 
variables.

A crucial step of SSA is the low-rank approximation 
of the so-called trajectory matrix, which will be 
described in the next section. The most often used 
tool for this is the singular value decomposition (SVD), 
which is however sensitive to the presence of outliers 
in the data. Several authors have proposed 
outlier-robust SSA methods by replacing the SVD by more 
robust versions. However, none of these low-rank 
approximations took the special diagonal structure of 
the trajectory matrix into account. The main 
contribution of our work is a new robust low-rank 
approximation method tailored to this situation. 

The paper is organized as follows. Section \ref{sec:met} 
briefly surveys existing work and introduces the new 
approach, called \textit{RObust Diagonalwise Estimation 
of SSA} (RODESSA), including its algorithm, 
implementation, and forecasting method. It is proved 
that each step of the algorithm reduces the objective 
function.
Section \ref{sec:out} proposes an enhanced time series 
plot in which two types of outliers are represented by 
colors, in order to assist with outlier detection. 
In Section \ref{sec:sim}, the performance of RODESSA is 
assessed by an extensive Monte Carlo simulation study. 
Section \ref{sec:exa} presents a real data example 
regarding temperature analysis in passenger railway 
vehicles. Section \ref{sec:con} concludes the paper.

\section{Multivariate singular spectrum analysis}
\label{sec:met}
\subsection{Classical multivariate SSA}
\label{sec:mssa}
Consider a $p$-variate time series 
$\mathbb{X}=\left(\mathbb{X}^{(1)}, \ldots, 
\mathbb{X}^{(p)}\right)$, i.e., a collection 
$\{\mathbb{X}^{(j)}=(x_i^{(j)})_{i=1}^N\,,$ 
$ j= 1, \ldots, p\}$ of $p$ time series 
of length $N$. Multivariate SSA then proceeds by the 
following four successive steps.
\vspace{3mm}

\noindent{\bf 1. Embedding.}
In the embedding step, the multivariate time series 
$\mathbb{X}$ is mapped into a big trajectory matrix 
$\bX$. Let $L$ be an integer called 
\textit{window length}, $1<L<N$. 
For each time series $\mathbb{X}^{(j)}$ we then form 
$K_u=N-L+1$ lagged vectors $X_i^{(j)}=\left(x_i^{(j)}, 
\ldots, x_{i+L-1}^{(j)}\right)^{T}$ 
for $1 \leqslant i \leqslant K_u$. 
The trajectory matrix of the multivariate series 
$\mathbb{X}$ is a matrix of size $L \times K$ 
with $K=p K_u$\,, and has the form
\begin{align} \label{eq:stack}
 \mathcal{T}_{\text {MSSA}}(\mathbb{X}) = 
 \bX &= \left[X_1^{(1)}: \ldots: 
 X_{K_u}^{(1)}: \ldots: X_1^{(p)}: \ldots:  
 X_{K_u}^{(p)}\right]=\left[\bX^{(1)}: 
 \ldots: \bX^{(p)}\right]\\
 &=\left[
 \begin{matrix}
 x_1^{(1)} & x_2^{(1)}&x_3^{(1)}&\ldots\\
 x_2^{(1)} & x_3^{(1)}&x_4^{(1)}&\ldots\\
 x_3^{(1)} & x_4^{(1)}&x_5^{(1)}&\ldots\\
 \vdots & \vdots & \vdots &\ddots \\
 \end{matrix}
 \hspace{0.3cm}
 \vline
 \hspace{0.3cm}
 \begin{matrix}
 x_1^{(2)} & x_2^{(2)}&x_3^{(2)}&\ldots\\
 x_2^{(2)} & x_3^{(2)}&x_4^{(2)}&\ldots\\
 x_3^{(2)} & x_4^{(2)}&x_5^{(2)}&\ldots\\
 \vdots & \vdots & \vdots & \ddots \\
 \end{matrix}
 \hspace{0.3cm}
 \vline
 \hspace{0.3cm}
 \begin{matrix}
 x_1^{(3)} & x_2^{(3)}&x_3^{(3)}&\ldots\\
 x_2^{(3)} & x_3^{(3)}&x_4^{(3)}&\ldots\\
 x_3^{(3)} & x_4^{(3)}&x_5^{(3)}&\ldots\\
 \vdots & \vdots & \vdots & \ddots \\
 \end{matrix}
 \hspace{0.3cm}
 \vline
 \hspace{0.3cm}
 \begin{matrix}
 \ldots\\
 \ldots\\
 \ldots\\
 \vdots\\
 \end{matrix}
 \right] \nonumber
\end{align}
\noindent where $\bX^{(j)}=
\mathcal{T}_{\text {SSA}}\left(\mathbb{X}^{(j)}\right)=
\left[X_1^{(j)}: \ldots: X_{K_u}^{(j)}\right]$ is 
the trajectory matrix of the one-dimensional series 
$\mathbb{X}^{(j)}$.
Note the (anti-)diagonal structure of each $\bX^{(j)}$, 
making it a so-called {\it Hankel matrix}. The entire 
matrix $\bX$ is thus a {\it stacked Hankel matrix}.
The notations $\mathcal{T}_{\text {SSA }}$ and 
$\mathcal{T}_{\text {MSSA }}$ stand for the univariate 
and multivariate embedding operators that map 
$\mathbb{X}^{(j)}$ and $\mathbb{X}$ to the 
corresponding trajectory matrices.\vspace{3mm}

\noindent{\bf 2. Decomposition.}
This step performs the SVD of the trajectory 
matrix $\bX$, yielding 
\begin{equation}\label{eq:SVD}
    \bX=\sum_{r=1}^{d}\beta_r\bu_r\bv_r^T
\end{equation}
where $d=\operatorname{rank}(\bX)$, and 
$\bu_1, \ldots, \bu_d$ and $\bv_1, \ldots, \bv_d$ 
are the left and right singular vectors. 
The ordered singular values are 
$\beta_1 \geqslant \cdots \geqslant \beta_d>0$, 
and the matrices $\bX_r=\beta_r \bu_r \bv_r^T$ 
all have rank 1. The triple 
$\left(\beta_r, \bu_r, \bv_r\right)$ is called 
the $r$-th eigentriple of the matrix $\bX$.\\

\noindent{\bf 3. Grouping.}
The grouping step corresponds to splitting the terms 
of~\eqref{eq:SVD} into several disjoint 
groups and summing the matrices within each group. 
In this paper we focus on the case where two main 
groups are created, and we write
\begin{equation} \label{eq:XhatPlusR}
   \bX= \bhX_\q + \bR
\end{equation}
where the estimate of the signal $\bhX_\q$ is a sum 
like~\eqref{eq:SVD} but only for the first $\q$ 
eigentriples,  whereas the residual matrix 
$\bR = \bX - \bhX_\q$ is associated with the noise. 
We can see $\bhX_\q$ as a low-rank approximation 
of the trajectory matrix.\vspace{3mm}

\noindent{\bf 4. Reconstruction.}
In this step, the fitted matrix $\bhX_\q$ is  
transformed back to the form of the input 
object $\mathbb{X}$ in~\eqref{eq:stack}.
In each submatrix $\bhX_\q^{(j)}$ we compute
average entries as follows. Denote an anti-diagonal 
as $A_i=\{(l, k): l+k=i+1, 
1 \leqslant l \leqslant L, 
1 \leqslant k \leqslant K_u\}$ with its 
cardinality $n_i=|A_i|$. Each matrix 
$\bhX_\q^{(j)} = (\hX_{lk}^{(j)})_{lk}$ is 
turned into a new series 
$\widehat{\mathbb{X}}^{(j)}=
(\hat{x_i}^{(j)})_{i=1,\ldots,N}$ 
of length $N$, with
\begin{equation} \label{eq_aver}
  \hat{x}_i^{(j)} := \frac{1}{n_i}
  \sum_{(l, k) \in A_i} 
  \hX^{(j)}_{l k}\;.
\end{equation}
The reconstructed multivariate time series is 
then given by $\widehat{\mathbb{X}}=
 \left(\widehat{\mathbb{X}}^{(1)}, 
 \ldots, \widehat{\mathbb{X}}^{(p)}\right)$.

\subsection{Outliers in multivariate time series}

Multivariate time series may contain outliers, that is, 
observations that deviate from the expected patterns or trends.
They can be caused by a variety of factors such as measurement 
errors, data entry mistakes, sensor malfunctions, or rare and 
unexpected events. Outliers can bias statistical measures, 
affect parameter estimation, and lead to inaccurate forecasting, 
resulting in erroneous conclusions and flawed decision-making. 
As in the multivariate setting 
\citep{alqallaf2009propagation,raymaekers2023challenges}, we
distinguish between a {\it cellwise outlier}, which is an 
outlying value $x_i^{(j)}$ in the $j$-th univariate time series
only, and a {\it casewise outlier}, where at a time $i$ several 
or all of the values $x_i^{(j)}$ deviate simultaneously. 

Figure~\ref{fig_met} illustrates this for a 3-variate time series.
The purple triangle is a cellwise outlier which affects only the 
first univariate time series at time $i=3$, whereas the orange 
squares indicate a casewise outlier at time point $i=5$. 
The effect of both types of outliers on the trajectory 
matrix~\eqref{eq:stack} is shown in the bottom part of the 
figure. The cellwise outlier corresponds to an antidiagonal in 
the leftmost Hankel matrix only, whereas the casewise outlier 
affects all three stacked Hankel matrices.

\begin{figure}[!ht]
\centering
\includegraphics[width=0.99\textwidth]{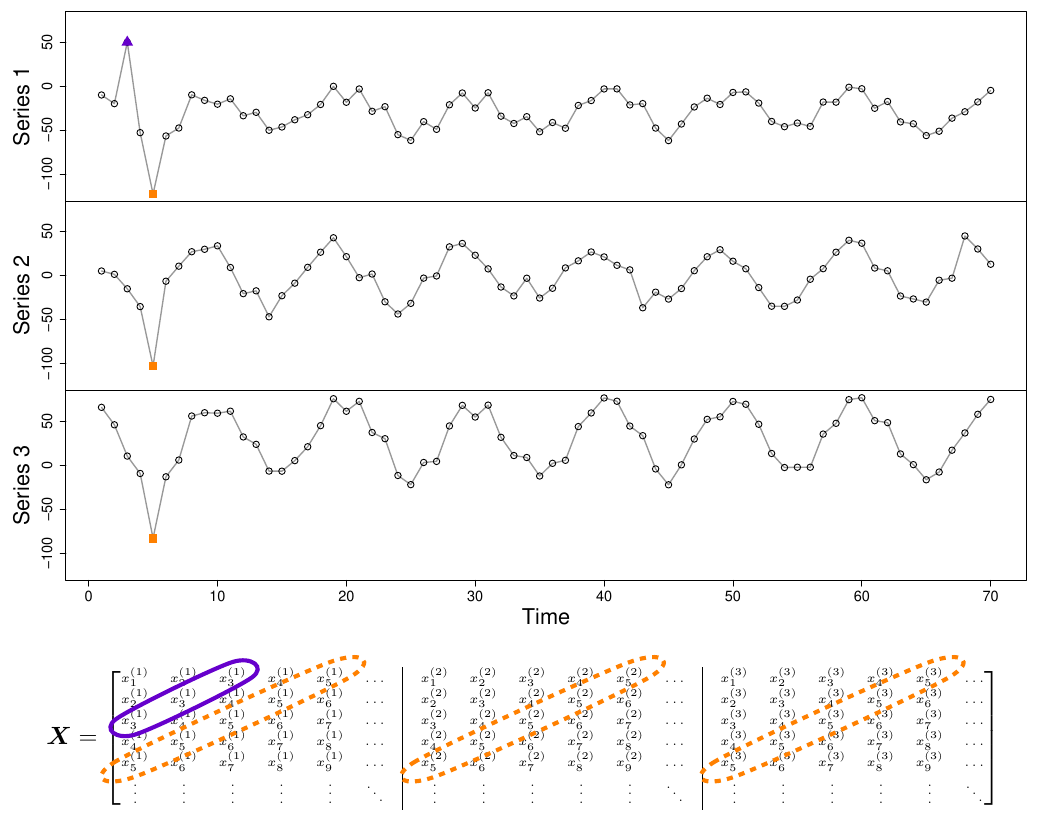}
\caption{The effect of a cellwise outlier (purple triangle) 
  and a casewise outlier (orange squares) on the diagonals 
  of the trajectory matrix $\bX$.}
\label{fig_met}
\end{figure}

\subsection{Existing robust SSA methods}
\label{sec:existing}

Classical MSSA, and its univariate version SSA, are 
sensitive to outliers because they are based on the SVD 
which is highly susceptible to outlying values.
In order to remedy this, there has been research into the
construction of outlier-robust SSA methods. In general, the
purpose of robust methods is to limit the effect of outliers 
on the results, after which the outliers can be detected by 
their residuals from the robust fit, see e.g.
\cite{RL1987}, \cite{maronna2019robust}. 

When constructing a more robust version of SSA one needs to
replace the classical SVD by something less sensitive
to outliers. One approach is to start from an outlier-robust 
principal component analysis (PCA) method. A PCA fit of rank
$\q$ is of the type
\begin{equation} \label{eq:PCA}
  \bhX = \bone_n \bhmu^T + \bU \bV^T 
\end{equation}
where $\bhmu$ is the estimated center, and the matrix of 
scores $\bU$ as well as the loadings matrix $\bV$ have
$\q$ columns. Many PCA methods have an option to set 
$\bhmu = \mathbf{0}$, and then~\eqref{eq:PCA} yields an 
approximation of rank $\q$ to $\bX$ as 
in~\eqref{eq:XhatPlusR}.
Afterward one can carry out the reconstruction step of SSA.
\cite{deklerk2015} applied the robust PCA method ROBPCA 
\citep{hubert2005robpca} to the trajectory matrix, and
then centered the original data $\bX$ as well as $\bhX$
by subtracting $\bone_n \bhmu^T$. In this way he
obtained a robust centered SSA method.

A limitation of this approach is that ROBPCA and most
other robust PCA methods are built to withstand
outlying rows of the data matrix, but not outlying cells. 
But we have seen in the bottom part of
Figure~\ref{fig_met} that a single outlying value in
a time series can affect many cells of the trajectory 
matrix $\bX$, especially when the window length $L$ is high.
Therefore a relatively small number of outliers in the
time series can affect over half the rows of $\bX$,
which ROBPCA might not withstand. In such
situations a cellwise robust PCA method like 
MacroPCA \citep{hubert2019macropca} could be used
instead.

In the nonrobust setting, expression~\eqref{eq:XhatPlusR} 
can equivalently be seen as a problem of low-rank 
approximation  of the trajectory matrix $\bX$ by 
a matrix $\widehat{\bX}_{L^2}$\,, which
minimizes
\begin{equation*}
   ||\bX-\bS||^2_{F}
\end{equation*}
over all $L \times K$ matrices $\bS$ of rank $\q$.
Writing the unknown $\bS$ as a 
product $\bS=\bU\bV^T$, where
$\bU=\left[\bu_1,\dots,\bu_\q\right]=
\left(\left[\bu^1,\dots,\bu^L\right]\right)^T$
is $L \times \q$ with
$\bu_i=\left(u_{i1},\dots,u_{iL}\right)^T$, and
$\bV=\left[\bv_1,\dots,\bv_\q\right]=
\left(\left[\bv^1,\dots,\bv^K\right]\right)^T$
is $K \times q$ with
$\bv_i=\left(v_{i1},\dots,v_{iK}\right)^T$, this is 
equivalent to minimizing 
\begin{equation} \label{eq_appF}
   ||\bX-\bU\bV^T||^2_{F}=
   \sum_{\ell=1}^{L}\sum_{k=1}^{K}
   \left(X_{\ell k}-\sum_{r=1}^\q u_{\ell r}v_{kr}
   \right)^2
   =\sum_{\ell=1}^{L}\sum_{k=1}^{K} R_{\ell k}^2
\end{equation}
with the residuals $R_{\ell k} := X_{\ell k}-
\sum_{r=1}^\q u_{\ell r}v_{kr}$\,. We then put 
$\widehat{\bX}_{L^2}:=
\widehat{\bU}_{L^2}\widehat{\bV}_{L^2}^T$.
The solution of the optimization~\eqref{eq_appF} is 
easily obtained through the SVD decomposition 
$\bX = \btU \bD \btV^T$ where $\bD$ is the diagonal matrix
of singular values. Restricting this to the $\q$ leading 
singular values 
$\beta_1 \geqslant \ldots \geqslant \beta_\q > 0$
we obtain $\bhX = \btU_\q \bD_q \btV_q^T$.
We can then absorb the singular values by putting
$\widehat{\bU}_{L^2} = \btU_\q \bD_q^{1/2}$
and $\widehat{\bV}_{L^2} = \btV_\q \bD_q^{1/2}$,
so that indeed $\widehat{\bX}=
\widehat{\bU}_{L^2}\widehat{\bV}_{L^2}^T$\,.
But the quadratic loss function in~\eqref{eq_appF} 
makes this a least squares fit, which is
very sensitive to outliers.

To remedy this, \cite{de2003framework} proposed a more
robust low-rank approximation by replacing~\eqref{eq_appF} 
by the minimization of the loss function
\begin{equation} \label{eq:DB}
  L_{\rho}\left(\bX-\bU\bV^T\right)
  :=\sum_{\ell=1}^{L}\sum_{k=1}^{K}\rho\left(
  \frac{X_{\ell k}-\sum_{r=1}^\q u_{\ell r}v_{kr}}
  {\hsigma}\right)
  =\sum_{\ell=1}^{L}\sum_{k=1}^{K}
  \rho\left(\frac{R_{\ell k}}{\hsigma}\right),
\end{equation}
where $\hsigma$ is a fixed scale estimate. 
The function $\rho$ must be continuous, even, 
non-decreasing for positive arguments, and 
satisfy $\rho(0)=0$. \cite{de2003framework} used
$$\rho(t) = \frac{t^2}{t^2 + 1}$$
that goes to 1 as $t \rightarrow \infty$. Therefore 
an outlying $t$ has much less effect on $\rho(t)$ 
than with the $\rho(t) = t^2$ in~\eqref{eq_appF}. 
They constructed an iterative algorithm for 
$\bhU \bhV^T$. For $\hsigma$ they took the median
absolute deviation of the residuals from 
an initial estimate $\bU_0\bV_0^T$.

Note that here and in the sequel the estimation target
is not the pair $(\bU,\bV)$ because that is not
uniquely defined. Indeed, if we take a nonsingular
$\q \times \q$ matrix $\bA$ we see that 
$(\bU\bA, \bV(\bA^{-1})^T)$ yields the same product
$\bU\bA\bA^{-1}\bV^T = \bU\bV^T$ as $(\bU,\bV)$,
and hence the same objective~\eqref{eq:DB}. The
actual estimation target is the product
$\bU\bV^T$\,.

\cite{chen2015robust} applied this general low-rank
approximation method to the decomposition step of SSA.
They replaced the function $\rho$ by 
Tukey’s biweight function 
\begin{equation} \label{eq_bis}
    \rho_c(t)=\begin{cases}
     \displaystyle 1 - \left(1 - \frac{t^2}{c^2}\right)^3 
      & |t|\leqslant c\\
      1 & |t|>c
    \end{cases}      
\end{equation}
with tuning constant $c = 4.685$, and used the method 
to filter seismic noise.

\cite{rodrigues2018robust} also constructed a robust SSA 
from the low-rank approximation method of 
\cite{de2003framework},
but replaced the function $\rho$ by $\rho(t)=|t|$ as in 
\cite{croux2003RAR}, yielding an $L^1$ objective. 
An advantage of the latter $\rho$ function is that 
$\hsigma$ can be moved out of~\eqref{eq:DB}, so one does 
not need to estimate $\sigma$ in advance. 
In subsequent work, \cite{rodrigues2020decomposition} 
carried out robust SSA based on the low-rank approximation
of \cite{zhang2013robust} which used the Huber function
\begin{equation}\label{eq:Huber}
  \rho_b(t) = \frac{t^2}{2} I(|t| \leqslant b) + 
            \left(b|t| - \frac{b^2}{2}\right)I(|t| > b)
\end{equation}
with $b=1.345$.
\cite{cheng2015application} performed a robust SSA 
analysis by applying a different robust low-rank 
approximation method due to \cite{candes2011robust}.

\subsection{The RODESSA objective function} \label{sec:obj}

The existing robust low-rank approximation methods described 
above are less sensitive to outliers than the classical SVD.
But none of them are tailored to the (possibly stacked)
Hankel structure of the trajectory matrix $\bX$. 
As we saw in the bottom part of Figure~\ref{fig_met},
an outlier in one of the time series corresponds to an
entire diagonal of the corresponding trajectory matrix. 
In the least squares low-rank approximation setting, 
algorithms have been developed that incorporate the Hankel 
structure of $\bX$ \citep{markovsky2008}, but no such
approach exists in the robust setting yet.

To fill this gap we propose a new robust method for MSSA, 
named RObust Diagonalwise Estimation of SSA (RODESSA), 
which explicitly takes into account the way outliers and
their residuals occur in the stacked Hankel 
structure of the trajectory matrix.
It is meaningful to talk about anomalous diagonals 
rather than anomalous rows of the trajectory matrix 
$\bX$. Therefore 
we propose to approximate the trajectory matrix $\bX$ 
by $\widehat{\bX}=\bU\bV^T$ obtained by minimizing
\begin{equation} \label{eq_rodi}
  L_{\rho_1,\rho_2}\left(\bX-\bU
  \bV^T\right) := 
  \sum_{i=1}^{N}p n_i \hsigma_2^2
  \rho_2\left(\frac{\sum_{j=1}^{p} 
  n_i\hsigma_{1,j}^2\rho_1\left(
  \sum_{a=1}^{n_i}(x_{i}^{(j)}-\hx_{ia}^{(j)})^2/
  (n_i\hsigma_{1,j}^2)\right)}{pn_i\hsigma_2^2}\right)
\end{equation}
over $(\bU,\bV)$ with $\rank(\bU\bV^T)=\q$.
Here $p$ is again the number of univariate time series
and $n_i$ is the length of the diagonal corresponding 
to $x_{i}^{(j)}$ as in \eqref{eq_aver}.
The predicted cell $\hx_{ia}^{(j)}$ is given by
$\sum_{r=1}^\q u_{i^*r}v_{a^*r}$ with 
$i^*=\min\lbrace L,i\rbrace +1-a$ and 
$a^*=\min\lbrace K_u,i\rbrace -(n_i-a)+K_u(j-1)$.
The fixed scales $\hsigma_{1,j}^2$ standardize the 
squared norms of the diagonal residuals of 
series $j$, given by
\begin{equation}\label{eq:r_ij}
  r_i^{(j)} := \frac{1}{n_i}\sum_{a=1}^{n_i}
  (x_{i}^{(j)} -\hx_{ia}^{(j)})^2\;.
\end{equation}
The overall scale $\hsigma_2^2$ standardizes 
the quantities
\begin{equation}\label{eq:r_i}
  r_i := \frac{1}{p} \sum_{j=1}^{p}\hsigma_{1,j}^2
  \rho_1\left(
  \frac{r_i^{(j)}}{\hsigma_{1,j}^2} \right)
\end{equation}  
based on all $p$ coordinates.
The functions $\rho_1$ and $\rho_2$ are defined for
nonnegative arguments and must be continuous and 
non-decreasing.
In our implementation both are of the form
$\rho(t) := \rho_c(\sqrt{t})$ where $\rho_c$ is 
Tukey's biweight~\eqref{eq_bis}.
[For $\rho_1(t) = \rho_2(t) = |t|$, 
\eqref{eq_rodi} would reduce to~\eqref{eq_appF}.] 
The goal of the normalization by $n_i\hsigma_{1j}^2$ 
inside $\rho_1$ in~\eqref{eq_rodi} is to give the
$r_i^{(j)}$ similar average sizes. Indeed, if the
$x_{i}^{(j)}-\hat{x}_{ia}^{(j)}$ 
would be independent normal random variables with 
variance $\hsigma_{1j}^2$\,, then $n_i r_i^{(j)}$ 
would follow a Gamma distribution with parameters 
$n_i/2$ and $2\hsigma_{1j}^2$ and thus with mean equal 
to $n_i\hsigma_{1j}^2$\,. Therefore the average of the
$r_i^{(j)}/\hsigma_{1j}^2$ would be close to~1. 
A similar argument applies for the normalization 
by $pn_i\hsigma_2^2$ inside $\rho_2$\,.

The functions $\rho_1$ and $\rho_2$ reduce the 
effect of cellwise and casewise outliers on the 
final estimates, because large values 
of $r_i^{(j)}$ and/or $r_i$ contribute less to the 
objective function. We can say that $r_i^{(j)}$ 
measures how prone the $i$-th value $x_i^{(j)}$ of 
time series $j$ is to be a cellwise outlier.
Analogously, $r_i$ reflects how prone the 
multivariate $\bx_i=(x_i^{(1)},\ldots,x_i^{(p)})^T$ 
is to be a casewise outlier.
Note that in the computation of $r_i$ the effect 
of cellwise outliers is tempered by the presence of 
$\rho_1$\,, to avoid that a single cellwise 
outlier would always result in a large 
casewise $r_i$\,.

\subsection{The IRLS algorithm}
\label{sec:alg}
We now address the optimization problem~\eqref{eq_rodi}.
The scale estimates $\hsigma_{1,j}$ and $\hsigma_2$ 
are constants, whose computation will be described
in Section~\ref{sec:implem}.
Because $L_{\rho_1,\rho_2}$ is continuously differentiable, 
its solution must satisfy the first-order necessary 
conditions for optimality.
They are obtained by setting the gradients of 
$L_{\rho_1,\rho_2}$ with respect to $\bu^1,\dots,\bu^L$ 
and $\bv^1,\dots,\bv^K$ to zero, yielding  
\begin{align}\label{eq_3b}
  \bV^T\bW^\ell(\bV\bu^\ell-\bX^\ell) &= \bzero, 
  \quad \ell=1,\dots,L,\nonumber\\
  \bU^T\bW_k(\bU\bv^k-\bX_k) &= \bzero, 
  \quad k=1,\dots,K,
\end{align}
where $\bX^1,\dots,\bX^L$ and $\bX_1,\dots,\bX_K$ are 
the rows and columns of $\bX$. 
Here $\bW^\ell$ is a $K \times K$ diagonal matrix, 
whose diagonal entries are equal to the $\ell$th row 
of the $L \times K$ weight matrix 
\begin{equation} \label{eq:updateW}
\bW=\lbrace w_{\ell k}\rbrace=\btW_c\odot\btW_r
\end{equation}
where the Hadamard product $\odot$ multiplies matrices 
entry by entry.
Analogously, $\bW_k$ is an $L \times L$ diagonal
matrix, whose diagonal entries are the $k$th
column of the matrix $\bW$.
The matrix $\btW_c$ in~\eqref{eq:updateW} is given by
$\btW_c=\left[\btW_c^{(1)}: \ldots: \btW_c^{(p)}\right]$ 
with $\btW_c^{(j)}=\mathcal{T}_{\mathrm{SSA}}
 \left(w_{c,1}^{(j)},\dots,w_{c,N}^{(j)}\right)$, 
containing the {\it cellwise weights}
\begin{equation}\label{eq:cellweight}
   w_{c,i}^{(j)}=\rho_1'\left(
   r_i^{(j)}/\hsigma^2_{1j}\right), 
   \quad i=1,\dots,N,\quad j=1\dots,p,
\end{equation}
in which $\rho_1'$ is the derivative of $\rho_1$\,.
The matrix $\btW_r$ is given by $\btW_r=
\left[\btW_r^{(1)}: \ldots: \btW_r^{(p)}\right]$,  
with each $\btW_r^{(j)}=\mathcal{T}_{\mathrm{SSA}}
 \left(w_{r,1},\dots,w_{r,N}\right)$ containing
the same {\it casewise weights}
\begin{equation}\label{eq:caseweight}
   w_{r,i}=\rho_2'\left(r_i/\hsigma^2_2\right), 
   \quad i=1,\dots,N\,.
\end{equation}
Outlying cells $x_i^{(j)}$ should get a small
cellwise weight $w_{c,i}^{(j)}$\, and outlying cases
$\bx_i$ should get a small casewise weight $w_{r,i}$\,.

The system~\eqref{eq_3b} is nonlinear because the weight 
matrices depend on the estimate, and the estimate
depends on the weight matrices. 
In such a situation one can resort to an iteratively 
reweighted least squares (IRLS) algorithm.
Note that for a fixed weight matrix $\bW$, 
the system~\eqref{eq_3b} coincides with the first-order 
necessary condition of the weighted least squares 
problem of minimizing
\begin{equation} \label{eq_4}
  \sum_{\ell=1}^{L}\sum_{k=1}^{K}w_{\ell k}
    \left(X_{\ell k}-\hX_{\ell k}\right)^2
\end{equation}
where $\hX_{\ell k}=\sum_{r=1}^\q u_{\ell r}v_{kr}$.
The optimization of~\eqref{eq_4} can be performed by 
alternating least squares \citep{gabriel1978least}. 
This minimizes~\eqref{eq_4} with respect to 
$\bm{u}^1,\dots,\bm{u}^L$ where $\bv^1,\dots,\bv^K$
and the weights are fixed, and alternates this 
with minimizing~\eqref{eq_4} with respect to
$\bv^1,\dots,\bv^K$ where $\bm{u}^1,\dots,\bm{u}^L$ 
and the weights are fixed. 

The overall algorithm starts from initial estimates 
$\bU_0,\bV_0$ of $\bU,\bV$ that will be described 
in Section~\ref{sec:implem}. New matrices 
$\bU_{t+1}$, $\bV_{t+1}$ are obtained from $\bU_t$, 
$\bV_t$ by
\begin{align} \label{eq_3c}
  \bv^k_{t+1}=
  \left(\bU_{t}^T\bW_{k,t}\bU_{t}\right)^{-1}
  \bU_{t}^T\bW_{k,t}\bX_k
  \quad &\mbox{for} \quad k=1,\dots,K,\nonumber\\
  \bu_{t+1}^\ell=
  \left(\bV_{t+1}^T\bW_{t}^\ell\bV_{t+1}\right)^{-1}
  \bV_{t+1}^T\bW^\ell_{t}\bX^\ell
  \quad &\mbox{for} \quad \ell=1,\dots,L.
\end{align}
Then the weight matrix is updated to $\bW_{t+1}$ as
in~\eqref{eq:updateW}. The iterative process continues 
until convergence, as summarized in Algorithm~\ref{al_1}.
\begin{algorithm}[H]
\caption{IRLS algorithm} \label{al_1}
\begin{algorithmic}[1]
  \STATE{Compute $\bU_0,\bV_0$, $\hsigma_{1,j}$ and 
     $\hsigma_2$ according to Section~\ref{sec:implem}}
     \Comment{\textit{Initialization} }
  \STATE{Set $t=0$}
  \STATE{Compute $\bW_{0}$ as in~\eqref{eq:updateW}}
  \REPEAT  
     \STATE{$\bv^k_{t+1}=
        \left(\bU_{t}^T\bW_{k,t}\bU_{t}\right)^{-1}
        \bU_{t}^T\bW_{k,t}\bX_k
        \quad \mbox{for} \quad k=1,\dots,K$}
     \STATE{$\bu_{t+1}^\ell= \left(
        \bV_{t+1}^T\bW_{t}^\ell\bV_{t+1}\right)^{-1}
        \bV_{t+1}^T\bW^\ell_{t}\bX^\ell
        \quad \mbox{for} \quad \ell=1,\dots,L$}
     \STATE{Compute $\bW_{t+1}$ as in~\eqref{eq:updateW}}
     \Comment{\textit{Weight matrix update} }
   \STATE{$t=t+1$}
  \UNTIL{$\;||\bU_t\bV_t^T - \bU_{t-1}\bV_{t-1}^T||_F
         < \nu\, ||\bU_{t-1}\bV_{t-1}^T||_F$ 
         \;for some tolerance $\nu$.} 
\end{algorithmic}
\end{algorithm}

\begin{theorem} \label{the_1}
Each iteration of Algorithm~\ref{al_1} decreases the 
objective function, that is,\linebreak 
$L_{\rho_1,\rho_2}(\bX - \bhX_{t+1}) \leqslant
 L_{\rho_1,\rho_2}(\bX - \bhX_t)$.
\end{theorem}
The proof is given in section A.1 of the Supplementary
Material. Since the objective function is decreasing 
and it has a lower bound of zero, it must converge.
Note that Proposition~\ref{al_1} is not restricted
to the functions $\rho_1$ and $\rho_2$ used in
RODESSA, which are of the type 
$\rho(t) = \rho_c(\sqrt{t})$ where $\rho_c$ is
Tukey's biweight \eqref{eq_bis}. All that is needed
is that the function $\rho(t)$ is differentiable and
concave. For this purpose $\rho_c$ could be replaced
by Huber's $\rho_b$ of \eqref{eq:Huber}, or the
function $\rho_{b,c}$ used in the wrapping
transform \citep{FROC2021}.

\subsection{Matters of implementation}
\label{sec:implem}
This section describes several implementation 
specifics: the initialization strategy to select 
$\bU_0$ and $\bV_0$ in the IRLS algorithm, the scale 
estimates $\hsigma_{1,j}$ and $\hsigma_2$\,, and how 
to select the loss function tuning constants, the 
rank $\q$, and the window length $L$.\vspace{3mm}

\noindent{\bf Scale estimates.}
Given initial estimates $\bU_0$ and $\bV_0$ (see below), 
the scale estimates $\hsigma_{1,j}^2$ and $\hsigma_2^2$ 
are computed as M-scales of the quantities  $r_i^{(j)}$  
from~\eqref{eq:r_ij} and the $r_i$ of~\eqref{eq:r_i}, 
all with respect to the fit $\bU_0\bV_0^{T}$.
A scale M-estimator of a univariate sample 
$\left(z_1,\dots,z_n\right)$
is the solution $\hsigma$ of the equation
\begin{equation}\label{eq:Mscale}
  \frac{1}{n} \sum_{i=1}^n 
  \rho\left(\frac{z_i}{\sigma}\right)=\delta
\end{equation}
where $0 < \delta < 1$. In our implementation we chose 
$\rho$ to be Tukey’s biweight $\rho_c$ of~\eqref{eq_bis}.
We set $\delta = 0.5$, which ensures a 50\% 
breakdown value, and $c= 1.548$, which is the 
solution of $\mathrm{E}[\rho_c(z)]=0.5$ for 
$z \sim \mathrm{N}(0,1)$, to obtain consistency at the 
normal model.\vspace{3mm}

\noindent{\bf Initialization.}
Since the objective function $L_{\rho_1,\rho_2}$ 
of~\eqref{eq_rodi} is not convex when $\rho_1$ are
$\rho_2$ are biweight functions, $\bU_0$ and $\bV_0$ 
should be carefully selected to avoid that the IRLS 
algorithm ends up in a nonrobust solution.
We compute several candidate initial fits,
and select the best one among them.
One candidate fit is given by the first $\q$ terms 
of the standard nonrobust SVD as in~\eqref{eq_appF},
yielding the rank-$\q$ matrix $\bhX_1$.
The second candidate $\bhX_2$ is the fit obtained
from \eqref{eq:DB} for $\rho(t)=|t|$, and the third 
candidate $\bhX_3$ is the fit of 
\cite{candes2011robust}. For each of these candidate
solutions $\bhX_k$ we apply the scale M-estimate
\eqref{eq:Mscale} to all the values $r_i^{(j)}$
given by~\eqref{eq:r_ij}.
Then we select the $\bhX_k$ with the lowest 
M-scale.\vspace{3mm}

\noindent{\bf Tuning constants.}
The functions $\rho_1$ and $\rho_2$ in~\eqref{eq_rodi}
use the biweight formula~\eqref{eq_bis}, by
$\rho_1(t)=\rho_{c_1}(\sqrt{t})$ and 
$\rho_2(t)=\rho_{c_2}(\sqrt{t})$.
Now we need to choose the tuning constants $c_1$ and 
$c_2$. Following \cite{aeberhard2021} we set 
these tuning constants using a measure of 
downweighting at the reference model.
The idea is to select tuning constants such that 
the average of the weights $w_{c,i}^{(j)}$  
of~\eqref{eq:cellweight} matches a target value 
$\delta_c$ and the average of the weights $w_{r,i}$
of~\eqref{eq:caseweight} matches a target 
value $\delta_r$\,.
For this computation the weights are standardized
to range from 0 to 1.
The target values determine how much $r_i^{(j)}$ and 
$r_i$ are downweighted on average, and by default we 
set $\delta_c = \delta_r = 0.9$\,.
The corresponding values of $c_1$ and $c_2$ are
obtained by Monte Carlo. This computation pretends
that the data are clean, with i.i.d.\ errors following
the standard normal distribution, and that the 
fitted $\bU\bV^T$ equals the true value.
We can then simulate the distribution of all the
weights $w_{c,i}^{(j)}$ and $w_{r,i}$ for the given 
values of $N$, $p$, and $L$, 
and choose $c_1$ and $c_2$.\vspace{3mm}

\noindent{\bf Selecting the rank $\q$.}
In classical MSSA, a popular way to select the rank
$\q$ is to make a plot of the unexplained variance.
This is the expression $||\bX-\bhX_r||^2_{F}$ as in 
\eqref{eq_appF}, where $\bhX_r$ is the best
approximation of $\bX$ of rank $r$. The unexplained
variance is monotone decreasing in $r$, and one
wants to select a value $\q$ where the curve has
an `elbow'. For the RODESSA method we 
inspect the analogous curve of the objective
function~\eqref{eq_rodi} at the rank-$r$ fit
$\bhX_r$\,, that is, we plot the curve of
\begin{equation*}
   L_{\rho_1,\rho_2}\left(\bX-\bhX_r\right)
\end{equation*}
versus $r$.\vspace{3mm}

\noindent{\bf Selecting the window length.} The choice of 
the window length $L$ in MSSA is a complex issue that has 
been addressed by \cite{hassani2013multivariate} and
\cite{golyandina2018singular}.
In general, $L$ should be selected to benefit either 
separability of the signal and the noise, or forecasting 
accuracy. Following \cite{golyandina2018singular} we use 
$L \simeq pN/(p+1)$  for the analysis of a small number 
$p$ of time series, and $L \simeq N/2$ otherwise.

\subsection{Forecasting with RODESSA}
\label{sec:for}

The MSSA model assumes a linear recurrent relation
\citep{golyandina2018singular}. This implies that
an observation can be predicted by a linear 
combination of the previous $L-1$ observations.
In classical MSSA the coefficients of this linear
combination are derived from the unique SVD fit
of rank $\q$ to the trajectory matrix $\bX$,
see e.g. \cite{danilov1997}.
For RODESSA the rank-$\q$ fit $\bhX_q$ can 
similarly be decomposed by the exact SVD, yielding
the $L \times \q$ matrix
$\btU=\left[\btu_1,\dots,\btu_\q\right]$ 
of left singular vectors 
and the $K \times \q$ matrix
$\btV=\left[\btv_1,\dots,\btv_\q\right]$ of 
right singular vectors.
Then the coefficient vector 
$\bha=\left(\ha_{L-1},\dots,\ha_{1}\right)^T$ is 
obtained as
\begin{equation*}
  \bha=\frac{\sum_{r=1}^\q\tilde{u}_{r,L}
       (\tilde{u}_{r,1},\ldots,\tilde{u}_{r,L-1})^T}
       {1-\sum_{r=1}^\q\tilde{u}_{r,L}^2}
\end{equation*}
where $\tilde{u}_{r,1},\ldots,\tilde{u}_{r,L}$
are the $L$ components of $\btu_r$\,. 
Let $(\hx_{i}^{(j)})_{i=1}^N$ be the reconstructed 
multivariate time series associated with the rank-$\q$
approximate trajectory matrix $\bhX_\q$.
Then the $h$-step ahead recurrent MSSA forecasts 
$\hx_{N+1}^{(j)},\dots,\hx_{N+h}^{(j)}$ for 
$j=1,\dots,p$ are given by
\begin{equation} \label{eq:forecast}
  \hx_i^{(j)}=\sum_{l=1}^{L-1}\ha_l \hx_{i-l}^{(j)}
  \quad \quad \mbox{for} \quad i=N+1,\dots,N+h\,.
\end{equation}

\section{Outlier detection}
\label{sec:out}

RODESSA implicitly provides information about outliers 
in multivariate time series. The cellwise weight 
$w_{c,i}^{(j)}$ given by~\eqref{eq:cellweight} reflects 
how much faith the algorithm has in the reliability of 
$x_i^{(j)}$, the value of the $j$-th time series at 
time $i$. The casewise weight $w_{r,i}$ given 
by~\eqref{eq:caseweight} does the same for the entire 
case $\bx_i = (x_i^{(1)},\ldots,x_i^{(p)})^T$ at time $i$. 
A small weight means that the corresponding cell or 
case was deemed less trustworthy, and was only
allowed to have a small effect on the fit. 

We propose a new graphical representation, 
called {\it enhanced time series plot}, which 
facilitates outlier detection by visualizing these 
weights in a single plot. Since the cellwise weights
$w_{c,i}^{(j)}$ are a monotone function of the
cellwise squared norms $r_i^{(j)}$ of~\eqref{eq:r_ij},
this can be seen as an extension of the cellmap
for multivariate data 
\citep{rousseeuw2018detecting,hubert2019macropca}
to time series.

We illustrate the enhanced time series plot on 
the publicly available Electricity Load Diagrams 
2011-2014 dataset 
\citep{misc_electricityloaddiagrams20112014_321}. 
It contains the electricity consumption of 
370 clients from January 2011 to December 2014. 
We plot the data of 4 clients from November 27th
to December 3rd 2014, with consumption observed 
every 2 hours. For the purpose of illustration 
we inserted some outliers. 

Figure \ref{fig_tsmap} shows its enhanced time 
series plot. The solid black lines correspond to 
the reconstructed multivariate time series.
The original cells $x_i^{(j)}$, connected 
by yellow solid lines, are represented by circles 
filled with a color. Those with cellwise weight 
$w_{c,i}^{(j)}$ close to~1 are filled with white.
Cells with a low weight and positive residual
are shown in increasingly intense red, and
those with a negative residual in increasingly
intense blue.
Moreover, $x_i^{(j)}$ is flagged as a cellwise 
outlier when $w_{c,i}^{(j)} < q_{c,\alpha}$ where 
$q_{c,\alpha}$ is the $\alpha$-quantile of the 
simulated distribution of cellwise weights at 
the reference model, described in 
Section~\ref{sec:implem}. Cells flagged as 
cellwise outliers are shown as solid squares,
in red when the residual is positive and in
blue when it is negative.

\begin{figure}[t]
\centering
\includegraphics[width=0.98\textwidth]
  {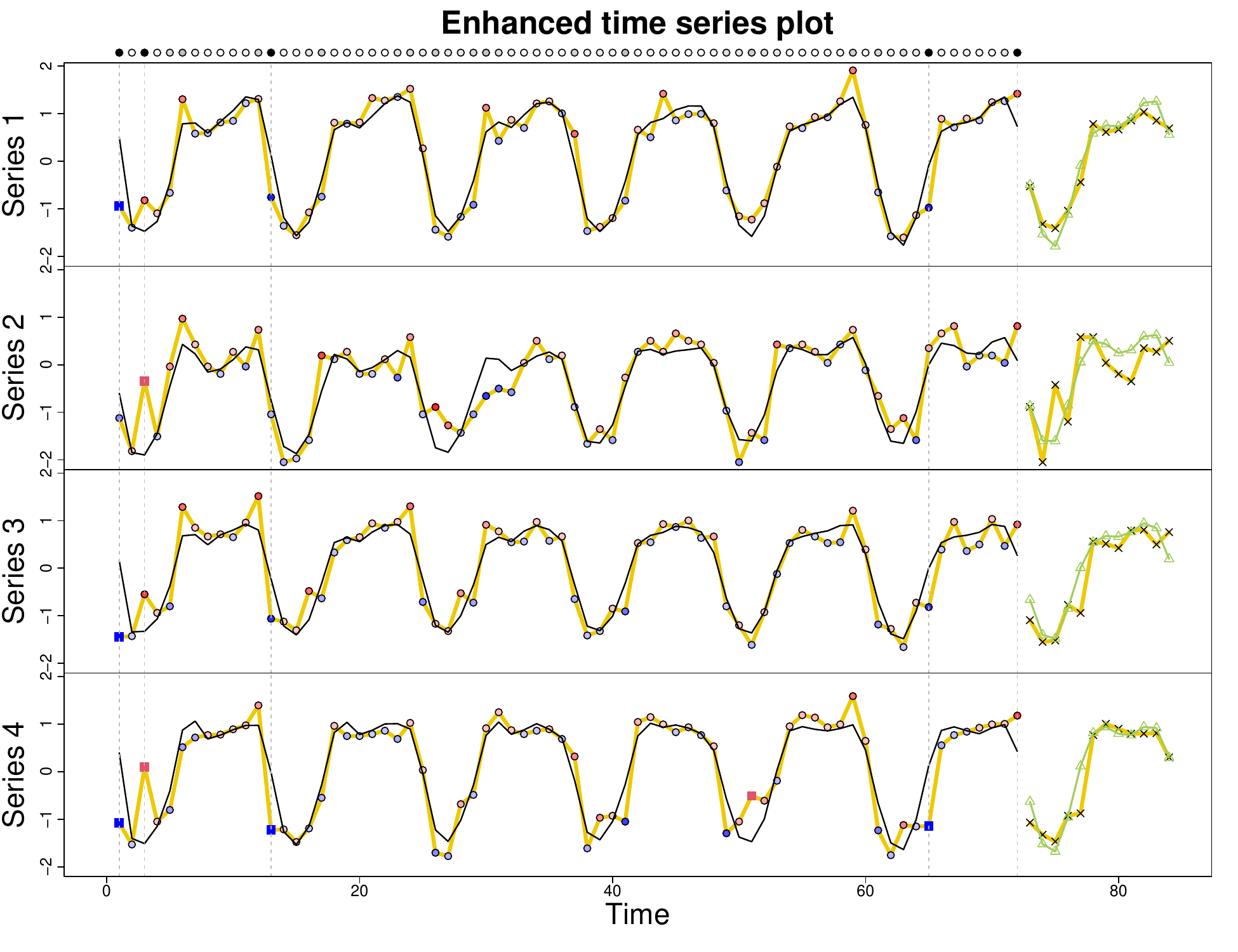}
\caption{Enhanced time series plot obtained by
applying RODESSA to the Electricity Load 
Diagrams data.}
\label{fig_tsmap}
\end{figure}

At the top of the plot we see circles 
reflecting the casewise weights $w_{r,i}$\,
with weight 1 shown with a  white interior
and lower weights in increasingly darker grey.
A case $\bx_i$ is flagged 
as a casewise outlier if $w_{r,i}$ is below
the $\alpha$-quantile $q_{r,\alpha}$ of the 
simulated distribution of casewise weights at 
the reference model. Casewise outliers are
indicated by a black solid circle and a 
vertical dashed grey line. The last one in 
the plot would not have been flagged
in the cellwise way.

In this example we also want to
plot the forecasts, which is not
part of the default enhanced time series plot.
The forecasts are shown as green 
triangles connected by green solid lines. 
The true values (which the algorithm did not
know about) were added as 
black crosses connected by yellow solid lines,
illustrating the forecasting performance.

\section{Simulation study}
\label{sec:sim}
In this section, the performance of the proposed 
RODESSA method in the presence of cellwise and 
casewise outliers is 
assessed through a Monte Carlo simulation study.
The data generation process is inspired by the 
simulated example in Section 3.3 of 
\cite{golyandina2018singular}.
Specifically, the multivariate time series is generated 
through the following signal plus noise model
\begin{equation*}
    x_i^{(j)}=  s^{(j)}(i) +\varepsilon_i \quad \quad 
    \text{for } i=1,\dots,N,\quad j=1,\dots,p\,,
\end{equation*}
with 
\begin{equation*}
    s^{(j)}(i)=A^{(j)}  \cos(2\pi i/ 10+C^{(j)})\,,
\end{equation*}
where $\varepsilon_i \sim N(0,\sigma^2)$, $\sigma=20$, 
$p=4$, and $N=70$. The constants $A^{(j)}$ and $C^{(j)}$ 
are set according to the three scenarios in 
Table~\ref{tab_sim}.

\begin{table}
\centering
\caption{Parameters  $ A^{(j)}$  and $ C^{(j)}$ 
  for scenarios 1, 2 and 3.}
\label{tab_sim}
\resizebox{0.5\textwidth}{!}{
\begin{tabular}{clcclcclcc}
\toprule
 &  & \multicolumn{2}{c}{Scenario 1}&  & 
 \multicolumn{2}{c}{Scenario 2}&  & 
 \multicolumn{2}{c}{Scenario 3}  \\ 
 \cline{1-1} \cline{3-4} \cline{6-7}\cline{9-10}
 $j$&& $ A^{(j)}$ &$ C^{(j)}$& &
  $ A^{(j)}$ &$ C^{(j)}$&&$ A^{(j)}$ &$ C^{(j)}$\\ 
 \cline{1-1} \cline{3-4} \cline{6-7}\cline{9-10}
 1&&20&0&&35&0&&20&0\\ 
 2&&30&0&&35&$\pi/ 5$&&30&$\pi/ 5$\\ 
 3&&40&0&&35&0&&40&0\\
 4&&50&0&&35&$\pi/ 5$&&50&$\pi/ 5$\\
\bottomrule
\end{tabular} }
\end{table}

In each scenario we consider cellwise and casewise 
contamination settings, with the fraction of 
outliers $\varepsilon$ equal to $0.1$ and $0.2$. 
In the cellwise contamination setting, outliers are 
introduced by adding the number $\gamma\sigma$ to 
$\varepsilon pN$ randomly chosen values $x_i^{(j)}$. 
We let $\gamma$ range over $\gamma=0,1,\dots,8$, 
so $\gamma=0$ corresponds with the uncontaminated 
setting.
In the casewise contamination setting, we add 
$\gamma\sigma$ to all values of $\varepsilon N$ 
randomly chosen 
$\bx_i = (x_i^{(1)},\dots,x_i^{(4)})^T$. 

Our proposed RODESSA method is compared with 
several competing approaches. We run the classical 
version of  MSSA (labeled CMSSA) and several robust 
versions described in Section~\ref{sec:existing}, 
which perform the decomposition step 
by~\eqref{eq:DB} with $\rho(t)=|t|$ as in 
\cite{rodrigues2018robust} (labeled RLM), as well 
as the method of \cite{cheng2015application} 
(labeled CHENG), and that of \cite{chen2015robust} 
(labeled CS). 

RODESSA is implemented as described in 
Section~\ref{sec:met} with $\q=2$, which 
corresponds to the true rank since for any 
angle $\varphi$ the signal 
$A \cos(2\pi i/10 + \varphi)$ equals the
linear combination\linebreak 
$A \cos(2\pi i/10)\cos(\varphi)  -
A \sin(2\pi i/10)\sin(\varphi)$ of the
functions $\cos(\varphi)$ and $\sin(\varphi)$.
Figure~\ref{fig:rank} plots the objective 
function for increasing rank, which also
indicates that most of the variability is 
explained by two components. We consider both 
$L =56 \simeq pN/(p+1)$ and $L=35=N/2$, and 
choose the tuning constants by setting 
$\delta_c=\delta_r=0.90$.

\begin{figure}[!ht]
\centering
\includegraphics[width=0.5\textwidth]
   {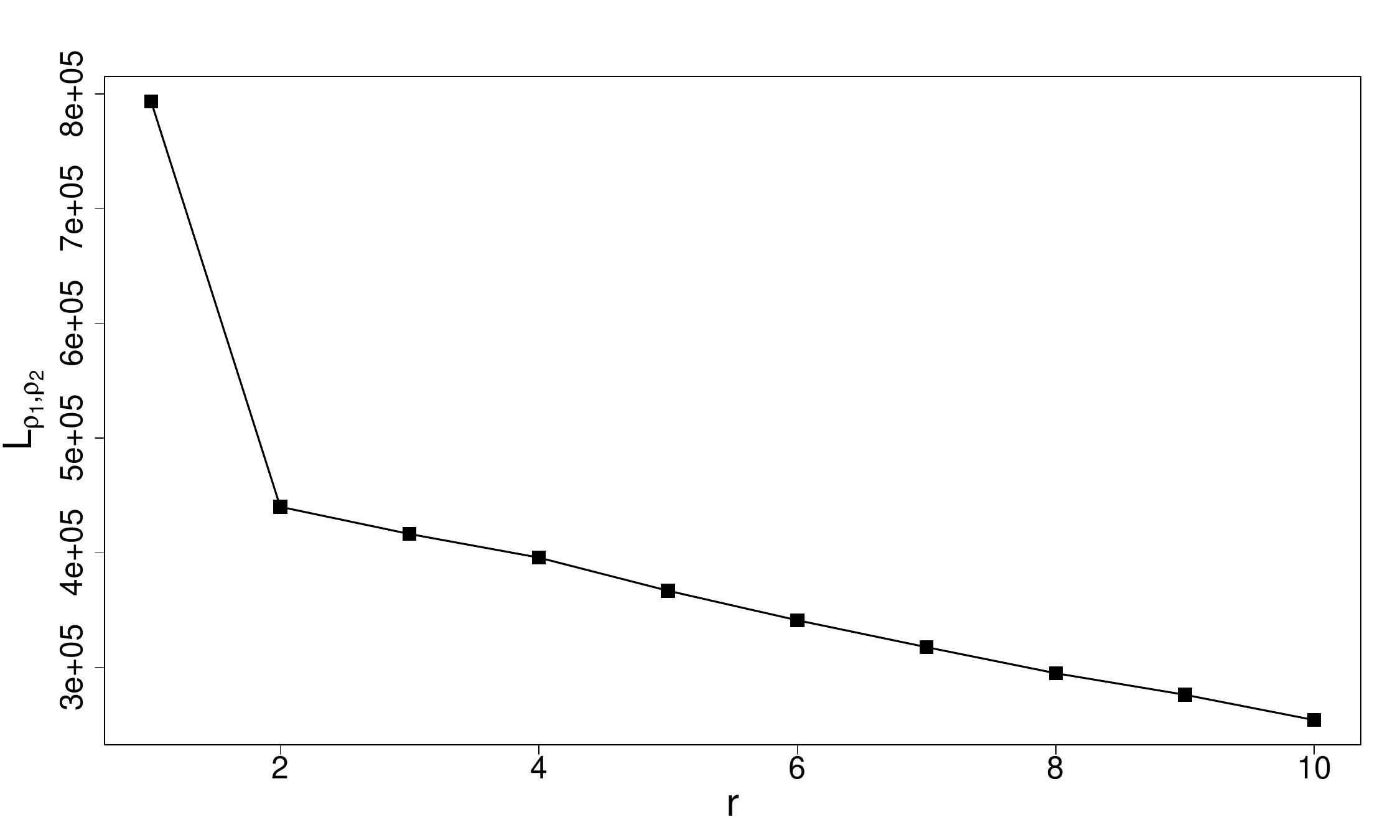}
\caption{Graph of the objective function for 
increasing values of the rank $r$.} 
\label{fig:rank}
\end{figure}

The competing approaches are implemented with the 
same values of $\q$ and $L$. To guarantee a fair 
comparison between RODESSA and CS, the tuning 
parameter in the CS method is chosen by means of 
the procedure in Section~\ref{sec:implem}. 
For each combination of scenario, contamination 
setting, magnitude $\gamma$ and outlier fraction, 
2000 replications are carried out.
The performance is assessed by means of the 
{\it reconstruction error} (RE) and the $20$-step 
ahead {\it forecasting error} (FE) that are 
defined as
\begin{equation}
\label{eq_fe}
  \text{RE}=\frac{1}{pN}\sum_{j=1}^{p}
  \sum_{i=1}^{N}
  \left(\hat{x}_i^{j}-s^{j}(i)\right)^2
  \quad \quad\ \quad
  \text{FE}=\frac{1}{20p}\sum_{j=1}^{p}
  \sum_{i=N+1}^{N+20}
  \left(\hat{x}_i^{j}-s^{j}(i)\right)^2\;.
\end{equation}
In the first formula the $\hat{x}_i^{j}$ are 
the reconstructions $(i \leqslant N)$. In the
second formula they are the forecasts ($i>N$) as 
defined in \eqref{eq:forecast}. The $\text{RE}$ 
and $\text{FE}$ are then averaged over the 
replications.

We report the results for the most challenging 
Scenario~3, with $L=N/2$. The supplementary 
material contains the simulation results for the 
other settings, which yield qualitatively 
similar conclusions. 

For cellwise contamination, 
Figure~\ref{fig:results_cell} shows the average 
RE and FE as a function of the contamination 
magnitude $\gamma$ for both outlier fractions. 
Without contamination ($\gamma=0$), CMSSA is the 
best method in terms of average RE and FE, as 
expected in this situation. But the errors of 
RODESSA and CS are similarly small, whereas 
those of CHENG and RLM are larger. 

\begin{figure}[!ht]
\centering
\includegraphics[width=0.73\textwidth]{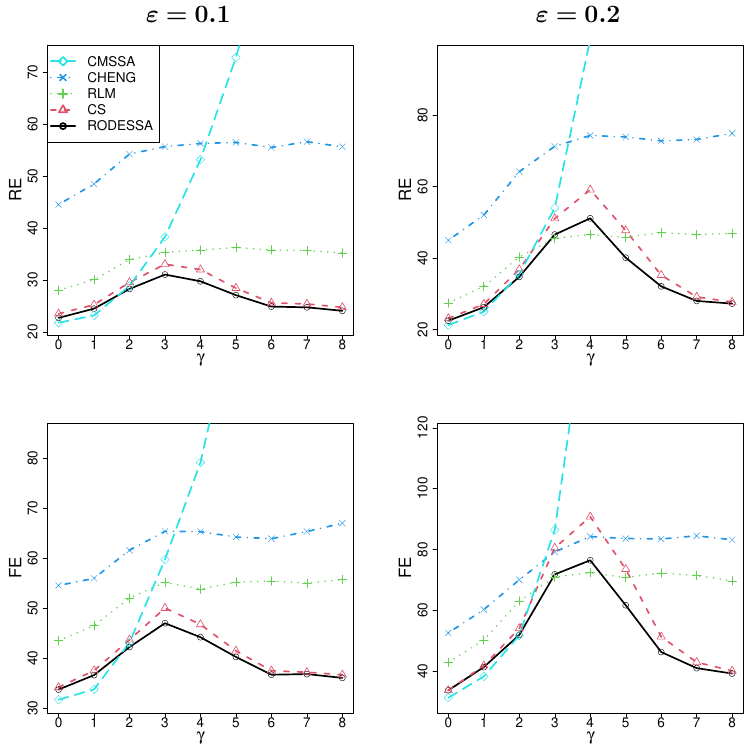}
\caption{Cellwise outliers. Average RE (top) and 
FE (bottom) attained by CMSSA, CHENG, RLM, CS, and 
RODESSA for Scenario 3, contamination probability  
$\varepsilon=0.1$ (left) and $\varepsilon=0.2$ 
(right), as a function of $\gamma$.}
\label{fig:results_cell}
\end{figure}

When looking at increasing $\gamma > 0$ it appears
that far outlying cells have an unbounded effect on 
CMSSA, a bounded effect on CHENG and RLM, and a small 
effect on RODESSA and CS, due to their bounded
biweight $\rho$ function. In that sense RODESSA and CS
are similar. But RODESSA does better than CS, due to
the fact that its decomposition step takes the
diagonal structure of the Hankel matrix into account.
The performance difference is largest for the higher 
outlier fraction $\varepsilon=0.2$.

\begin{figure}[!ht]
\centering
\includegraphics[width=0.73\textwidth]{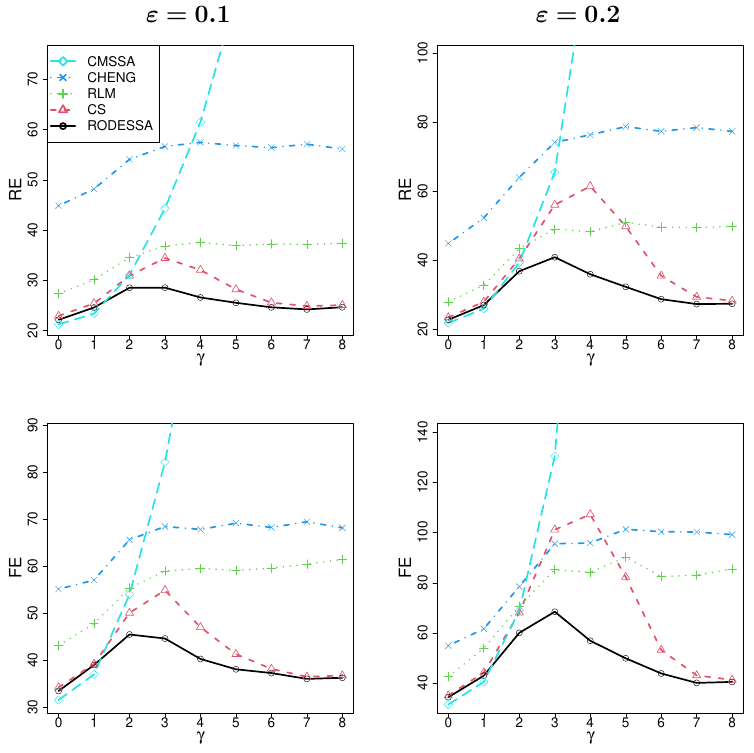}
\caption{Casewise outliers. Average RE (top) and 
FE (bottom) attained by CMSSA, CHENG, RLM, CS, and 
RODESSA for Scenario 3, contamination probability  
$\varepsilon=0.1$ (left) and $\varepsilon=0.2$ 
(right), as a function of $\gamma$.}
\label{fig:results_row}
\end{figure}

Figure \ref{fig:results_row} compares the same
methods in the presence of casewise outliers.
Here the differences are larger, with RODESSA
outperforming the other methods more strongly.
This is due to the fact that RODESSA can share
information about diagonals across the stacked
structure. This helps because we saw in 
Figure~\ref{fig_met} that casewise outliers 
affect several diagonals in the trajectory matrix 
simultaneously. Unlike the other methods, RODESSA 
combines information across such diagonals through
the loss function $\rho_2$ in \eqref{eq_rodi}, 
which makes it more robust for casewise outliers.

\section{A real data example: temperature  
analysis in passenger railway vehicles}
\label{sec:exa}
In this section, a real data example from the 
railway industry illustrates the applicability 
and potential of RODESSA.

In recent years, railway transportation in Europe 
has emerged as a viable alternative to other 
modes of transport, leading to intense competition 
between operators to enhance passenger satisfaction 
\citep{kallas2011white}. A particular challenge in 
this context is ensuring optimal thermal comfort 
inside passenger rail coaches \citep{ye2004thermal}. 
To address this challenge, new European standards, 
such as UNI EN 14750 \citep{en200614750}, have been 
developed. These standards focus on regulating air 
temperature, relative humidity, air speed, and 
overall comfort and air quality in passenger rail 
coaches, taking into account the diverse operating 
requirements of trains. Consequently, railway 
companies are proactively installing sensors to 
gather and analyze data from on-board heating, 
ventilation, and air conditioning (HVAC) systems 
\citep{homod2013review}. HVAC systems play a crucial 
role in maintaining passenger thermal comfort, and 
their performance is being improved based on the 
insights obtained from the collected data.

We analyze operational data from HVAC systems 
installed on a passenger 6-coach train operating 
during the summer season \citep{lepore2022neural}.
The data are publicly available at 
\url{https://github.com/unina-sfere/NN4MSP}\,.
Specifically, the inside temperature of each coach 
was recorded about every four minutes from 10:00  
to 22:00, yielding $N=176$ observations of a 
multivariate time series with $p=6$ components.
We applied RODESSA in its default form as described 
in Section~\ref{sec:implem}, with 
$L = 151 \simeq pN/(p+1)$. We selected $\q=7$ based
on the values of the objective function. This
yields the enhanced time series plot shown in
Figure~\ref{fig_tsdata}.

\begin{figure}[t]
\centering
\includegraphics[width=\textwidth]
  {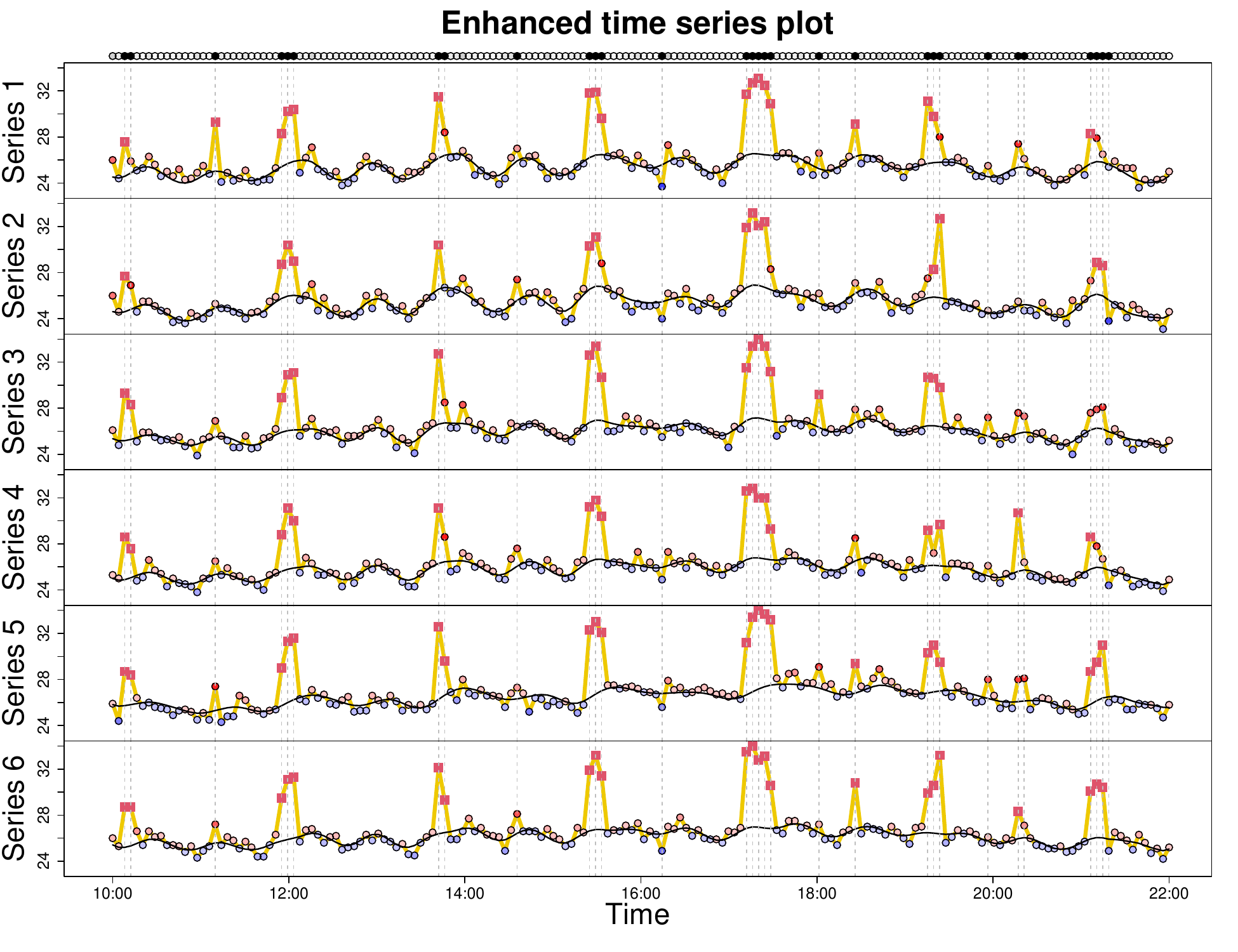}
\caption{Enhanced time series plot from applying 
RODESSA to the multivariate time series of 
temperatures.} \label{fig_tsdata}
\end{figure}

The vertical dashed lines in Figure~\ref{fig_tsdata}
indicate casewise outliers. It seems that the 
temperature inside the six coaches is significantly 
higher than expected in an almost periodic fashion.
This behavior is particularly severe for the group 
of casewise outliers between 17:00 and 18:00.
Discussions with domain experts revealed that those 
measurements were acquired during train stops at 
terminal stations. In this situation the coach doors 
were left open, and due to the summer season this
increased the temperature inside the coaches. 
From this we can conclude that those measurements 
are not representative of the operating conditions 
of the train and are not useful to characterize 
temperature dynamics.
Moreover, the reconstructed time series shows that 
the temperature inside coaches does increase and 
decrease periodically, which is in accordance with 
the on/off control of the HVAC system.
Note that not all casewise outliers are labeled as 
cellwise outliers. For instance, the 
casewise outlier shortly after 16:00 has no component 
that is flagged as a cellwise outlier, but the
temperatures were relatively low in all six coaches 
simultaneously. The ability to detect such effects
is a feature of the proposed method.

To assess the forecasting performance of RODESSA 
relative to the competing methods described in 
Section~\ref{sec:sim}, we used subsequences of 
the multivariate time series.
Specifically, starting from the first 120 
observations, $h$-step ahead forecasts of all 
methods are computed for $h=5,10,20$. 
Similarly, forecasts are obtained by considering 
the first 121, 122,... observations and so on.
For each subsequence, the $h$-step ahead 
forecasts are compared to the observed values of
the time series, by computing the median of their 
absolute differences, denoted as mFE.
Note that here we do not use the forecasting error 
formula in \eqref{eq_fe} because the data we are
predicting contains outliers as well.
 
Moreover, to assess the effect of the 
window length $L$ on the forecasting performance 
we compare the two lengths given in 
Section \ref{sec:implem}, namely $1/2$ and $6/7$ 
of the subsequence length $N_{sub}$\,.
Figure~\ref{fig_box} shows boxplots of mFE for 
$h=5,10,20$ and $L=N_{sub}/2$ (top) as well as 
$L=6N_{sub}/7$ (bottom).
The RODESSA method outperforms its competitors
overall, which is in line with the simulation 
study. 

\begin{figure}[ht]
\centering
\includegraphics[width=0.99\textwidth]
  {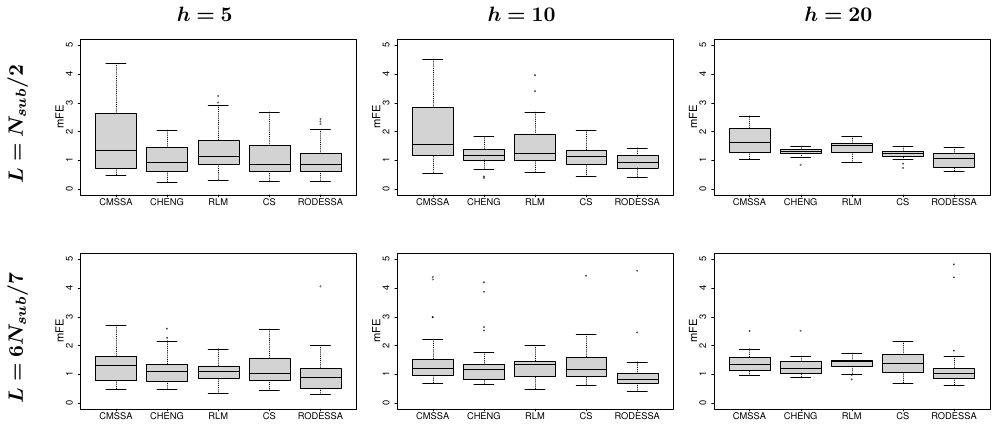}
\caption{Boxplots of mFE for $h=5$, $10$, $20$ 
and $L=N_{sub}/2$ (top) as well as 
$L=6N_{sub}/7$ (bottom).}
\label{fig_box}
\end{figure}

\section{Conclusions}
\label{sec:con} 

Multivariate singular spectrum analysis (MSSA) is a 
technique for fitting vector time series and
making forecasts. It starts by constructing a matrix
with a special structure. It consists of several
submatrices stacked side by side, one for each component
of the multivariate time series. A cell of the time 
series, that is, the value of one of its coordinates
at a given timepoint, corresponds to a diagonal in the
corresponding submatrix. A case of the time series, that
is, all of the coordinates at a given time point,
corresponds to diagonals in all of the submatrices.

Sometimes cells are outlying, and sometimes entire
cases. But the classical MSSA method can be strongly 
affected by outliers, because it is based on a 
low-rank singular value decomposition, which is a 
least squares fit. Several more robust methods
have been proposed in the literature, based on 
versions of the SVD that are less sensitive to 
outliers. However, these versions do not take the 
diagonal structure into account, and even a small
number of outliers can create a large number of 
outlying diagonal entries that affect many rows
and columns of the matrix, thereby overwhelming
the fit. To resolve this issue we propose the
RODESSA method, which explicitly takes the diagonal 
structure into account when decomposing the matrix.
This makes it more robust than its predecessors,
as illustrated in the extensive simulation study
reported in Section 4 and the Supplementary
Material. Moreover, it loses little
efficiency when there are no outliers.

The RODESSA method is performed by a fast
algorithm based on iteratively reweighted
least squares. In Section 2 we prove a 
proposition stating that each step of the 
algorithm decreases the objective function. 
We also propose a new graphical
display, called the enhanced time series plot,
which visualizes the weights of the cells as well
as the cases, making outliers stand out.
We have applied the RODESSA method to a real 
multivariate time series about temperatures in 
passenger railway vehicles. This illustrates
its good forecasting performance, as well as
the convenient outlier detection by the enhanced 
time series plot.


\paragraph{Software availability.}
The R code and example scripts, and 
the data of the example in Section~\ref{sec:exa},
are publicly available on the webpage
\url{https://wis.kuleuven.be/statdatascience/robust}\,.

\paragraph{Funding Details.}
This work was supported by the Flanders Research 
Foundation (FWO) under Grant for a scientific stay 
in Flanders V505623N; and by Piano Nazionale di 
Ripresa e Resilienza (PNRR) - Missione 4 
Componente 2, Investimento 1.3-D.D. 1551.11-10-2022, 
PE00000004 within the Extended Partnership MICS 
(Made in Italy - Circular and Sustainable).

\spacingset{1}

\begin{thebibliography}{}

\bibitem[\protect\citeauthoryear{Aeberhard, Cantoni, Marra, and
  Radice}{Aeberhard et~al.}{2021}]{aeberhard2021}
Aeberhard, W., E.~Cantoni, G.~Marra, and R.~Radice (2021).
\newblock Robust fitting for generalized additive models for location, scale
  and shape.
\newblock {\em Statistics and Computing\/}~{\em 31}, 11.

\bibitem[\protect\citeauthoryear{Alqallaf, Van~Aelst, Yohai, and
  Zamar}{Alqallaf et~al.}{2009}]{alqallaf2009propagation}
Alqallaf, F., S.~Van~Aelst, V.~J. Yohai, and R.~H. Zamar (2009).
\newblock Propagation of outliers in multivariate data.
\newblock {\em The Annals of Statistics\/}~{\em 37}, 311--331.

\bibitem[\protect\citeauthoryear{{British Standards Institution}}{{British
  Standards Institution}}{2006}]{en200614750}
{British Standards Institution} (2006).
\newblock {BS-EN} 14750: Railway applications -- {A}ir conditioning for urban
  and suburban rolling stock. {P}art 1: Comfort parameters.

\bibitem[\protect\citeauthoryear{Broomhead and King}{Broomhead and
  King}{1986}]{broomhead1986qualitative}
Broomhead, D. and G.~King (1986).
\newblock On the qualitative analysis of experimental dynamical systems.
\newblock In S.~Sarkar (Ed.), {\em Nonlinear Phenomena and Chaos}, pp.\
  113--144. Hilger Ltd.

\bibitem[\protect\citeauthoryear{Cand\`{e}s, Li, Ma, and Wright}{Cand\`{e}s
  et~al.}{2011}]{candes2011robust}
Cand\`{e}s, E.~J., X.~Li, Y.~Ma, and J.~Wright (2011).
\newblock Robust principal component analysis?
\newblock {\em Journal of the ACM\/}~{\em 58\/}(3), 1--37.

\bibitem[\protect\citeauthoryear{Chen and Sacchi}{Chen and
  Sacchi}{2015}]{chen2015robust}
Chen, K. and M.~D. Sacchi (2015).
\newblock Robust reduced-rank filtering for erratic seismic noise attenuation.
\newblock {\em Geophysics\/}~{\em 80\/}(1), V1--V11.

\bibitem[\protect\citeauthoryear{Cheng, Chen, and Sacchi}{Cheng
  et~al.}{2015}]{cheng2015application}
Cheng, J., K.~Chen, and M.~D. Sacchi (2015).
\newblock Application of {R}obust {P}rincipal {C}omponent analysis {(RPCA)} to
  suppress erratic noise in seismic records.
\newblock In {\em SEG Technical Program Expanded Abstracts 2015}, pp.\
  4646--4651. Society of Exploration Geophysicists.

\bibitem[\protect\citeauthoryear{Croux, Filzmoser, Pison, and Rousseeuw}{Croux
  et~al.}{2003}]{croux2003RAR}
Croux, C., P.~Filzmoser, G.~Pison, and P.~J. Rousseeuw (2003).
\newblock Fitting multiplicative models by robust alternating regressions.
\newblock {\em Statistics and Computing\/}~{\em 13}, 23--36.

\bibitem[\protect\citeauthoryear{Danilov}{Danilov}{1997}]{danilov1997}
Danilov, D. (1997).
\newblock Principal components in time series forecast.
\newblock {\em Journal of Computational and Graphical Statistics\/}~{\em
  6\/}(1), 112--121.

\bibitem[\protect\citeauthoryear{De~Klerk}{De~Klerk}{2015}]{deklerk2015}
De~Klerk, J. (2015).
\newblock Time series outlier detection using the trajectory matrix in singular
  spectrum analysis with outlier maps and {ROBPCA}.
\newblock {\em South African Statistical Journal\/}~{\em 49}, 61--76.

\bibitem[\protect\citeauthoryear{De~la Torre and Black}{De~la Torre and
  Black}{2003}]{de2003framework}
De~la Torre, F. and M.~J. Black (2003).
\newblock A framework for robust subspace learning.
\newblock {\em International Journal of Computer Vision\/}~{\em 54}, 117--142.

\bibitem[\protect\citeauthoryear{Gabriel}{Gabriel}{1978}]{gabriel1978least}
Gabriel, K.~R. (1978).
\newblock Least squares approximation of matrices by additive and
  multiplicative models.
\newblock {\em Journal of the Royal Statistical Society: Series B
  (Methodological)\/}~{\em 40\/}(2), 186--196.

\bibitem[\protect\citeauthoryear{Golyandina, Korobeynikov, and
  Zhigljavsky}{Golyandina et~al.}{2018}]{golyandina2018singular}
Golyandina, N., A.~Korobeynikov, and A.~Zhigljavsky (2018).
\newblock {\em Singular spectrum analysis with {R}}.
\newblock Springer.

\bibitem[\protect\citeauthoryear{Golyandina, Nekrutkin, and
  Zhigljavsky}{Golyandina et~al.}{2001}]{golyandina2001analysis}
Golyandina, N., V.~Nekrutkin, and A.~Zhigljavsky (2001).
\newblock {\em Analysis of time series structure: SSA and related techniques}.
\newblock CRC press.

\bibitem[\protect\citeauthoryear{Golyandina and Zhigljavsky}{Golyandina and
  Zhigljavsky}{2013}]{golyandina2013singular}
Golyandina, N. and A.~Zhigljavsky (2013).
\newblock {\em Singular Spectrum Analysis for Time Series}.
\newblock Springer Berlin, Heidelberg.

\bibitem[\protect\citeauthoryear{Hassani and Mahmoudvand}{Hassani and
  Mahmoudvand}{2013}]{hassani2013multivariate}
Hassani, H. and R.~Mahmoudvand (2013).
\newblock Multivariate singular spectrum analysis: A general view and new
  vector forecasting approach.
\newblock {\em International Journal of Energy and Statistics\/}~{\em 1\/}(01),
  55--83.

\bibitem[\protect\citeauthoryear{Homod}{Homod}{2013}]{homod2013review}
Homod, R.~Z. (2013).
\newblock Review on the {HVAC} system modeling types and the shortcomings of
  their application.
\newblock {\em Journal of Energy\/}~{\em 2013}, 1--10.

\bibitem[\protect\citeauthoryear{Hubert, Rousseeuw, and Van~den Bossche}{Hubert
  et~al.}{2019}]{hubert2019macropca}
Hubert, M., P.~J. Rousseeuw, and W.~Van~den Bossche (2019).
\newblock Macro{PCA}: an all-in-one {PCA} method allowing for missing values as
  well as cellwise and rowwise outliers.
\newblock {\em Technometrics\/}~{\em 61\/}(4), 459--473.

\bibitem[\protect\citeauthoryear{Hubert, Rousseeuw, and Vanden~Branden}{Hubert
  et~al.}{2005}]{hubert2005robpca}
Hubert, M., P.~J. Rousseeuw, and K.~Vanden~Branden (2005).
\newblock {ROBPCA}: a new approach to robust principal component analysis.
\newblock {\em Technometrics\/}~{\em 47}, 64--79.

\bibitem[\protect\citeauthoryear{Kallas}{Kallas}{2011}]{kallas2011white}
Kallas, S. (2011).
\newblock {\em White Paper on transport: Roadmap to a single European transport
  area: towards a competitive and resource-efficient transport system}.
\newblock Office for Official Publications of the European Communities.

\bibitem[\protect\citeauthoryear{Lepore, Palumbo, and Sposito}{Lepore
  et~al.}{2022}]{lepore2022neural}
Lepore, A., B.~Palumbo, and G.~Sposito (2022).
\newblock Neural network based control charting for multiple stream processes
  with an application to {HVAC} systems in passenger railway vehicles.
\newblock {\em Applied Stochastic Models in Business and Industry\/}~{\em
  38\/}(5), 862--883.

\bibitem[\protect\citeauthoryear{Markovsky}{Markovsky}{2008}]{markovsky2008}
Markovsky, I. (2008).
\newblock Structured low-rank approximation and its applications.
\newblock {\em Automatica\/}~{\em 44}, 891--909.

\bibitem[\protect\citeauthoryear{Maronna, Martin, Yohai, and
  Salibi{\'a}n-Barrera}{Maronna et~al.}{2019}]{maronna2019robust}
Maronna, R.~A., R.~D. Martin, V.~J. Yohai, and M.~Salibi{\'a}n-Barrera (2019).
\newblock {\em Robust Statistics: {T}heory and Methods (with R)}.
\newblock John Wiley \& Sons.

\bibitem[\protect\citeauthoryear{Raymaekers and Rousseeuw}{Raymaekers and
  Rousseeuw}{2021}]{FROC2021}
Raymaekers, J. and P.~J. Rousseeuw (2021).
\newblock Fast robust correlation for high-dimensional data.
\newblock {\em Technometrics\/}~{\em 63}, 184--198.

\bibitem[\protect\citeauthoryear{Raymaekers and Rousseeuw}{Raymaekers and
  Rousseeuw}{2023}]{raymaekers2023challenges}
Raymaekers, J. and P.~J. Rousseeuw (2023).
\newblock Challenges of cellwise outliers.
\newblock {\em Econometrics and Statistics\/}~{\em X}, to appear.

\bibitem[\protect\citeauthoryear{Rodrigues, Louren{\c{c}}o, and
  Mahmoudvand}{Rodrigues et~al.}{2018}]{rodrigues2018robust}
Rodrigues, P.~C., V.~Louren{\c{c}}o, and R.~Mahmoudvand (2018).
\newblock A robust approach to singular spectrum analysis.
\newblock {\em Quality And Reliability Engineering International\/}~{\em
  34\/}(7), 1437--1447.

\bibitem[\protect\citeauthoryear{Rodrigues, Pimentel, Messala, and
  Kazemi}{Rodrigues et~al.}{2020}]{rodrigues2020decomposition}
Rodrigues, P.~C., J.~Pimentel, P.~Messala, and M.~Kazemi (2020).
\newblock The decomposition and forecasting of mutual investment funds using
  singular spectrum analysis.
\newblock {\em Entropy\/}~{\em 22\/}(1), 83.

\bibitem[\protect\citeauthoryear{Rousseeuw and Leroy}{Rousseeuw and
  Leroy}{1987}]{RL1987}
Rousseeuw, P.~J. and A.~Leroy (1987).
\newblock {\em Robust {R}egression and {O}utlier {D}etection}.
\newblock Wiley.

\bibitem[\protect\citeauthoryear{Rousseeuw and Van~den Bossche}{Rousseeuw and
  Van~den Bossche}{2018}]{rousseeuw2018detecting}
Rousseeuw, P.~J. and W.~Van~den Bossche (2018).
\newblock Detecting deviating data cells.
\newblock {\em Technometrics\/}~{\em 60\/}(2), 135--145.

\bibitem[\protect\citeauthoryear{Trindade}{Trindade}{2015}]{misc_electricityloaddiagrams20112014_321}
Trindade, A. (2015).
\newblock {Electricity Load Diagrams 20112014}.
\newblock UCI Machine Learning Repository.
\newblock {DOI}: https://doi.org/10.24432/C58C86.

\bibitem[\protect\citeauthoryear{Ye, Lu, Li, Sun, and Liu}{Ye
  et~al.}{2004}]{ye2004thermal}
Ye, X., H.~Lu, D.~Li, B.~Sun, and Y.~Liu (2004).
\newblock Thermal comfort and air quality in passenger rail cars.
\newblock {\em International Journal of Ventilation\/}~{\em 3\/}(2), 183--192.

\bibitem[\protect\citeauthoryear{Zhang, Shen, and Huang}{Zhang
  et~al.}{2013}]{zhang2013robust}
Zhang, L., H.~Shen, and J.~Z. Huang (2013).
\newblock Robust regularized singular value decomposition with application to
  mortality data.
\newblock {\em The Annals of Applied Statistics\/}~{\em 7}, 1540--1561.

\end{thebibliography}


\clearpage
\pagenumbering{arabic}
\appendix
\begin{center}
\large{Supplementary Material to: 
  Multivariate Singular Spectrum Analysis\\ by  
  Robust Diagonalwise Low-Rank Approximation}\\
\vspace{7mm}
\normalsize{Fabio Centofanti, Mia Hubert, 
  Biagio Palumbo, Peter J. Rousseeuw} 
\end{center}
\vspace{3mm}

\numberwithin{equation}{section} 
\renewcommand{\theequation}
  {A.\arabic{equation}} 

\spacingset{1.45} 

\section*{\large A.1\;\; Proof of Proposition~\ref{the_1}}
\label{sec:proof}

In this section Proposition~\ref{the_1} is proved, which 
ensures that each step of the algorithm decreases the
RODESSA objective function \eqref{eq_rodi}. In order to
prove Proposition~\ref{the_1}, we first need two
lemmas.

We will denote a potential fit as 
$\btheta = \bhX = \bU\bV^T$ where $\btheta$ 
belongs to the set of all $L \times K$ matrices
of rank at most $\q$.
We introduce the notation $\bof(\btheta)$ for the
column vector 
$\left( f_1(\btheta),\dots,f_{LK}(\btheta)\right)^T$
with $LK$ entries, which are the values 
$(\bX_{\ell k} - \btheta_{\ell k})^2$ for 
$\ell = 1,\ldots,L$ and $k = 1,\ldots,K$.
We can then write the RODESSA objective function 
\eqref{eq_rodi} as $L(\bof(\btheta)) := 
L_{\rho_1,\rho_2}(\bX-\btheta)$. 

\begin{lemma}
\label{le_1}
The function $\bof \rightarrow L(\bof)$ is concave.
\end{lemma}
\begin{proof}
We first show that the univariate function
$\rho: \mathbb R_{+} \rightarrow \mathbb R_{+}:
 t \rightarrow \rho_c(\sqrt{t})$, in which $\rho_c$
is Tukey's biweight function \eqref{eq_bis},
is concave. From  
$\rho(t) = 1 - (1 - t/c^2)^3 
I(0 \leqslant t \leqslant c^2)$
we can compute the second derivative
$$\rho''(t) \;=\; \frac{-6}{c^4}
\left(1-\frac{t}{c^2}\right)
I(0 \leqslant t \leqslant c^2) 
\;\leqslant\; 0$$
from which concavity follows. The functions
$\rho_1$ and $\rho_2$ in RODESSA are thus
concave.

By the definition of concavity of a multivariate 
function, we now need to prove that for 
any column vectors $\bof,\bol$ in 
$\mathbb R_{+} ^{LK}$ and any
$\lambda$ in $\left(0,1\right)$ it holds that  
$L(\lambda\bof+(1-\lambda)\bol) \geqslant 
\lambda L(\bof)+(1-\lambda)L(\bol)$.
This works out as
\begin{gather*}
 L(\lambda\bof+(1-\lambda)\bol)=\\
 \sum_{i=1}^{N}p n_i \hsigma_2^2\rho_2\left(\frac{\sum_{j=1}^{p} 
  n_i\hsigma_{1,j}^2\rho_1\left(
 \sum_{a=1}^{n_i}(\lambda f_{ia}^{(j)}+(1-\lambda)g_{ia}^{(j)})/
  (n_i\hsigma_{1,j}^2)\right)}{pn_i\hsigma_2^2}\right)\geqslant\\
  \sum_{i=1}^{N}p n_i \hsigma_2^2\rho_2\left(\frac{\sum_{j=1}^{p} 
  n_i\hsigma_{1,j}^2\left[\lambda\rho_1\left(
  \sum_{a=1}^{n_i}f_{ia}^{(j)}/
  (n_i\hsigma_{1,j}^2)\right)+(1-\lambda)\rho_1\left(\sum_{a=1}^{n_i}g_{ia}^{(j)}/
  (n_i\hsigma_{1,j}^2)\right)\right]}{pn_i\hsigma_2^2}\right)=\\
\end{gather*}
\begin{gather*}  
  \sum_{i=1}^{N}p n_i \hsigma_2^2\rho_2\left(\frac{\lambda\sum_{j=1}^{p} 
  n_i\hsigma_{1,j}^2\rho_1\left(
  \sum_{a=1}^{n_i}f_{ia}^{(j)}/
  (n_i\hsigma_{1,j}^2)\right)+(1-\lambda)\sum_{j=1}^{p} 
  n_i\hsigma_{1,j}^2\rho_1\left(\sum_{a=1}^{n_i}g_{ia}^{(j)}/
  (n_i\hsigma_{1,j}^2)\right)}{pn_i\hsigma_2^2}\right)\geqslant\\
  \hspace{-6cm}\sum_{i=1}^{N}p n_i \hsigma_2^2\left[\lambda\rho_2\left(\frac{\sum_{j=1}^{p} 
  n_i\hsigma_{1,j}^2\rho_1\left(
  \sum_{a=1}^{n_i}f_{ia}^{(j)}/
  (n_i\hsigma_{1,j}^2)\right)}{pn_i\hsigma_2^2}\right)\right.\\
  \left.\hspace{5cm} +\,(1-\lambda)\rho_2\left(\frac{\sum_{j=1}^{p} 
  n_i\hsigma_{1,j}^2\rho_1\left(\sum_{a=1}^{n_i}g_{ia}^{(j)}/
  (n_i\hsigma_{1,j}^2)\right)}{pn_i\hsigma_2^2}\right)\right]=\\
   \lambda L(\bof)+(1-\lambda)L(\bol).
\end{gather*}
The first inequality derives from the concavity of $\rho_1$
and the fact that $\rho_2$ is nondecreasing. The second 
inequality is due to the concavity of $\rho_2$\,. 
Therefore $L$ is a concave function.
\end{proof}

We can also write the weighted least squares objective
\eqref{eq_4} as a function of $\bof$. We will
denote it as $L_{\bW}(\bof) := (\vecmat(\bW))^T \bof$.
Here $\vecmat(.)$ turns a matrix into a column vector, 
in the same way as was done to obtain the column 
vector $\bof$. The next lemma makes a connection
between $L_{\bW}$ and the RODESSA objective $L$.

\begin{lemma} \label{lem_2}
If two column vectors $\bof,\bol$ in 
$\mathbb R_{+} ^{LK}$ satisfy
$L_{\bW}(\bof) \leqslant L_{\bW}(\bol)$, then 
$L(\bof) \leqslant L(\bol)$. 
\end{lemma}

\begin{proof}
From Lemma \ref{le_1} we know that $L(\bof)$ is 
concave as a function of $\bof$, and it is also 
differentiable because $\rho_1$ and $\rho_2$ are.
Therefore
\begin{equation*}
    L(\bof) \leqslant L_(\bol) +
    (\nabla L(\bol))^T(\bof-\bol)
\end{equation*}
where the column vector $\nabla L(\bol)$ is the 
gradient of $L$ in $\bol$. But $\nabla L(\bol)$ 
equals $\vecmat(\bW)$ by\linebreak construction, 
so $(\nabla L(\bol))^T(\bof-\bol) =
L_{\bW}(\bof) - L_{\bW}(\bol) \leqslant 0$.
Therefore $L(\bof) \leqslant L(\bol)$.
\end{proof}
With this preparation we can prove
Proposition \ref{the_1}.

\begin{proof}[Proof of Proposition \ref{the_1}]
When we go from $\btheta_t$ to $\btheta_{t+1}$ 
the least squares fits in steps 5 and 6\linebreak 
of Algorithm 1 ensure that 
$L_{\bW_t}(\bof(\btheta_{t+1})) \leqslant 
L_{\bW_t}(\bof(\btheta_t))$, so it follows from
Lemma 2 that\linebreak 
$L(\bof(\btheta_{t+1})) \leqslant 
L(\bof(\btheta_t))$. 
\end{proof}

\section*{\large A.2\;\; Additional 
          simulation results}

In this section we present all the simulation 
results according to the settings described in 
Section~\ref{sec:sim}.  
In the uncontaminated setting ($\gamma = 0$),
Figure~\ref{fig:dataic1} shows examples of  
multivariate time series generated by 
scenarios 1, 2 and 3. 
Scenario 1 generates series in which the 
components have increasing amplitude, whereas 
Scenario 2 shifts their phase. Scenario 3 
combines both effects. 

\begin{figure}[!ht]
\centering
\vspace{5mm}
\includegraphics[width=0.87\textwidth]
  {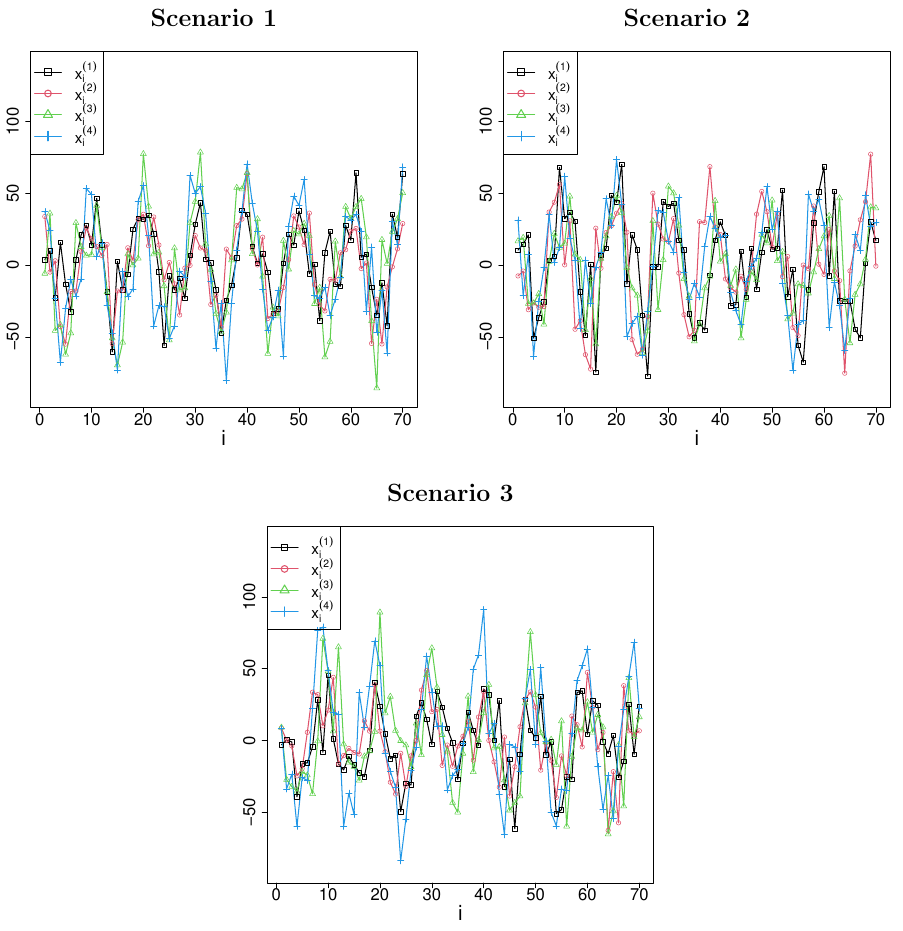}
\caption{(No outliers.) Examples of randomly generated 
multivariate time series under scenarios 1, 2, and 3 
of the simulation study. } \label{fig:dataic1}
\end{figure}
\clearpage

We first consider the window length $L=N/2=35$.
Figure~\ref{fig:results_cell_L2} shows the 
average RE and FE for each scenario (S1, S2 and
S3) and outlier fraction $\varepsilon$ as a 
function of $\gamma$, for data with cellwise 
contamination. 
Figure~\ref{fig:results_row_L2} does the same 
for casewise contamination. 

\begin{figure}[ht!]
\centering
\vskip0.8cm
\includegraphics[width=1.0\textwidth]
  {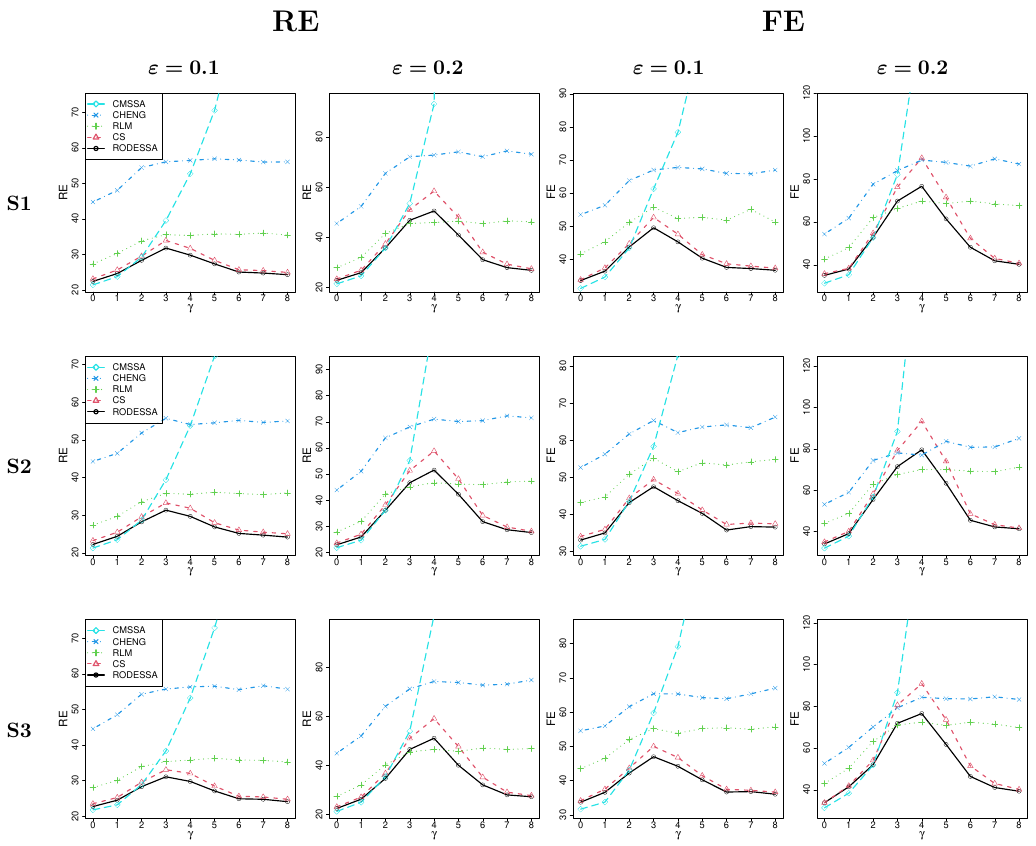}
\caption{(Window length $L=35$, cellwise outliers.) 
  Average RE and FE attained by CMSSA, CHENG, RLM, 
	CS, and  RODESSA for each scenario 
	(S1, S2, and S3) and contamination probability 
	($\varepsilon=0.1,0.2$) in function of $\gamma$.}
\label{fig:results_cell_L2}
\end{figure}

\begin{figure}[ht!]
\centering
\vskip-0.4cm
\includegraphics[width=1.0\textwidth]
  {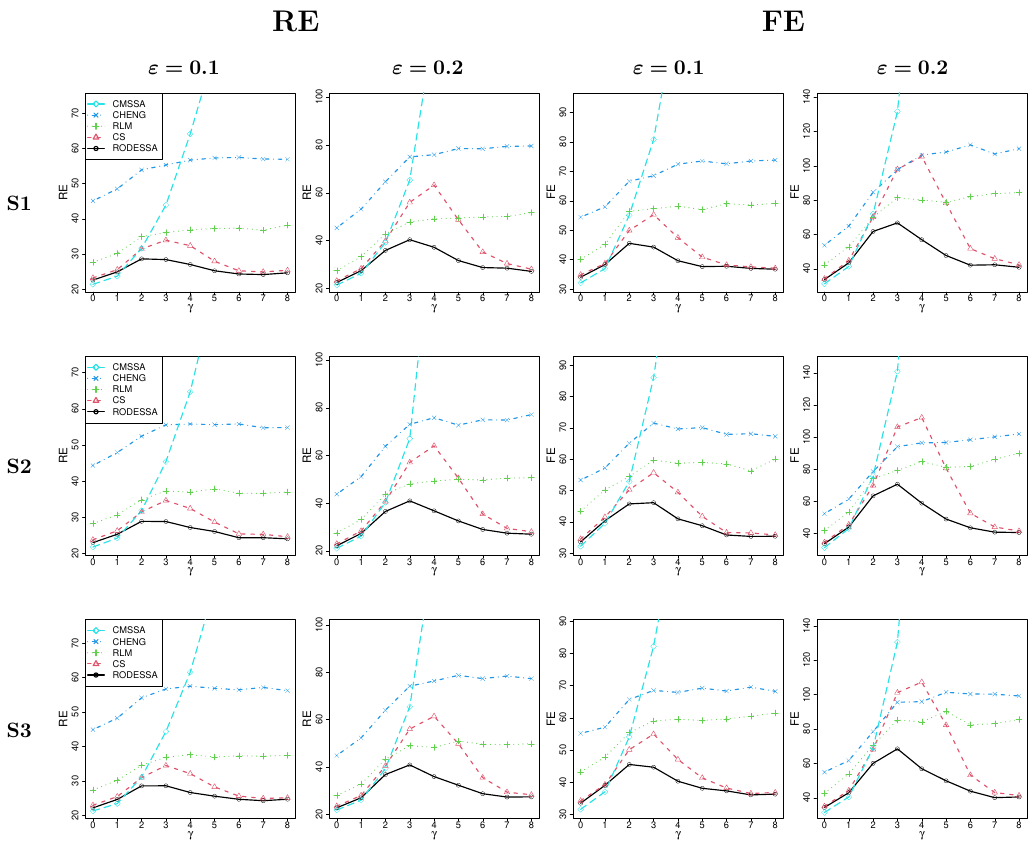}
\caption{(Window length $L=35$, casewise outliers.) 
  Average RE and FE attained by CMSSA, CHENG, RLM,
	CS, and  RODESSA  for each scenario 
	(S1, S2, and S3) and contamination probability 
	($\varepsilon=0.1,0.2$) in function of $\gamma$.}
\label{fig:results_row_L2}
\end{figure}

\clearpage
Figures~\ref{fig:results_cell_L1} 
and~\ref{fig:results_row_L1} show the corresponding
results for window length $L = 56 \sim pN/(p+1)$.

We conclude that all these simulation results are 
qualitatively similar to those reported in 
Section~\ref{sec:sim} of the paper.

\begin{figure}[ht!]
\centering
\vskip0.8cm
\includegraphics[width=1.0\textwidth]
  {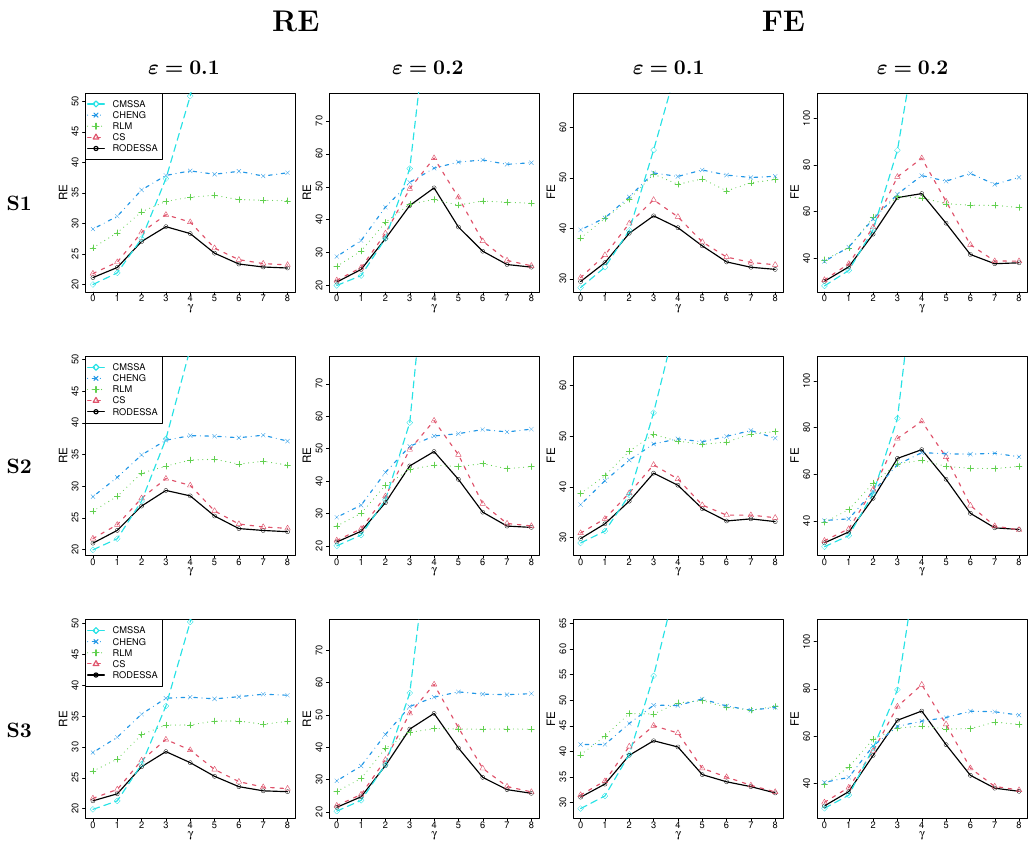}
\caption{(Window length $L=56$, cellwise outliers.) 
  Average RE and FE attained by CMSSA, CHENG, RLM, 
	CS, and RODESSA for each scenario 
	(S1, S2, and S3) and contamination probability 
	($\varepsilon=0.1,0.2$) in function of $\gamma$.}
\label{fig:results_cell_L1}
\end{figure}

\begin{figure}[ht!]
\centering
\vskip-0.4cm
\includegraphics[width=1.0\textwidth]
  {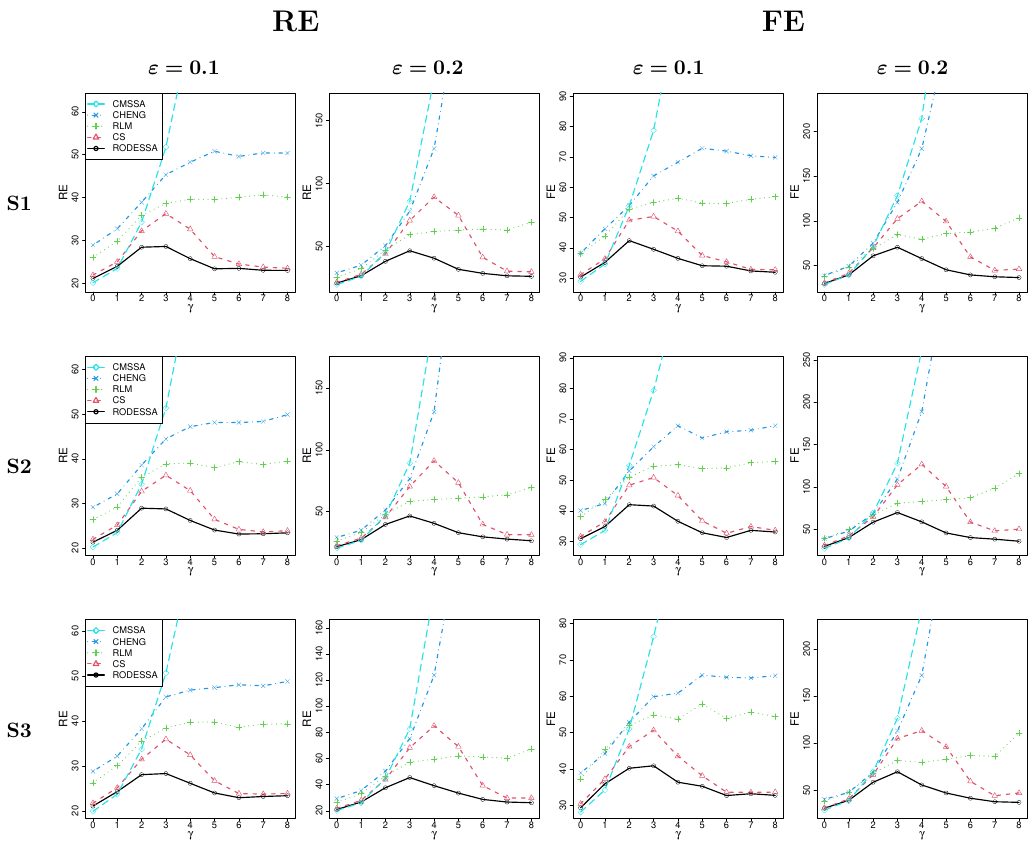}
\caption{(Window length $L=56$, casewise outliers.) 
  Average RE and FE attained by CMSSA, CHENG, RLM, 
	CS, and  RODESSA  for each scenario 
	(S1, S2, and S3) and contamination probability 
	($\varepsilon=0.1,0.2$) in function of $\gamma$.}
\label{fig:results_row_L1}
\end{figure}

\end{document}